%% file: main.tex
\documentclass{article}
\pdfoutput=1

\usepackage{main-macros}

\title{Primes via Zeros: Interactive Proofs for Testing Primality of Natural Classes of Ideals}

\author{
Abhibhav Garg
\thanks{University of Waterloo, Cheriton School of Computer Science. \email{\{a65garg, rafael\}@uwaterloo.ca}}
\and 
Rafael Oliveira~\orcidlink{0000-0001-8917-8689}
\samethanks
\and 
Nitin Saxena~\orcidlink{0000-0001-6931-898X}
\thanks{IIT Kanpur, Department of CSE. \email{nitin@cse.iitk.ac.in}}
}

\date{}

\begin{document}

\maketitle

\begin{abstract}
A central question in mathematics and computer science is the question of determining whether a given ideal $I$ is prime, which geometrically corresponds to the zero set of $I$, denoted $Z(I)$, being irreducible.
The case of principal ideals (i.e., $m=1$) corresponds to the more familiar absolute irreducibility testing of polynomials, where the seminal work of (Kaltofen 1995) yields a randomized, polynomial time algorithm for this problem.
However, when $m > 1$, the complexity of the primality testing problem seems much harder.
The current best algorithms for this problem are only known to be in EXPSPACE.    

Such drastic state of affairs has prompted research on the primality testing problem (and its more general variants, the primary decomposition problem, and the problem of counting the number of irreducible components) for natural classes of ideals.
Notable classes of ideals are the class of radical ideals, complete intersections (and more generally Cohen-Macaulay ideals).
For radical ideals, the current best upper bounds are given by (B\"{u}rgisser \& Scheiblechner, 2007), putting the problem in $\PSPACE$.
For complete intersections, the primary decomposition algorithm of (Eisenbud, Huneke, Vasconcelos 1992) coupled with the degree bounds of (DFGS 1991), puts the ideal primality testing problem in EXP.  
In these situations, the only known complexity-theoretic lower bound for the ideal primality testing problem is that it is $\coNP$-hard for the classes of radical ideals, and equidimensional Cohen-Macaulay ideals.

In this work, we significantly reduce the complexity-theoretic gap for the ideal primality testing problem for the important families of ideals $I$ (namely, {\em radical ideals} and {\em equidimensional Cohen-Macaulay ideals}).
For these classes of ideals, assuming the Generalized Riemann Hypothesis, we show that primality testing lies in $\Sigma_3^p \cap \Pi_3^p$.
This significantly improves the upper bound for these classes, approaching their lower bound, as the primality testing problem is $\coNP$-hard for these classes of ideals.

Another consequence of our results is that for equidimensional Cohen-Macaulay ideals, we get the first PSPACE algorithm for primality testing, exponentially improving the space and time complexity of prior known algorithms. 
\end{abstract}

\newpage
\setcounter{tocdepth}{2}
\tableofcontents

\newpage

\section{Introduction}\label{section: introduction}
\input{introduction-conference}
\subsection{Acknowledgements}
 N.S.~is grateful to the University of Waterloo for an invite to the distinguished lecture series in Jun'24. 
The discussions done in Waterloo's conducive environment led to this work. 
In addition, N.S. thanks the funding support from DST-SERB (JCB/2022/57) and N.Rama Rao Chair.
A.G.~would like to thank TIFR for very generously hosting him during the later part of this project.
Part of this work took place while R.O.~and N.S.~participated in the workshop on "Algebraic and Analytic Methods in Computational Complexity" in Dagstuhl (Sep'24). 
The authors would like to thank Schloss Dagstuhl for providing excellent hospitality and conditions to conduct research discussions.
The authors would like to thank Peter B\"{u}rgisser and Pascal Koiran for helpful conversations about this project during the workshop at Dagstuhl.

\section{Preliminaries}\label{section: preliminaries}

\input{preliminaries}

\section{Height bounds}\label{section: height bounds}
\input{height-alg-int}

\section{Geometric Irreducibility and base change}\label{section: base change}

\input{base-change-alg-int}

\section{Interactive proofs of primality testing}\label{section: interactive proofs}

\input{interactive-proofs}
\section{Conclusion \& open problems}\label{section: conclusion}

 \input{conclusion}

 \printbibliography

\end{document}

%% file: introduction-conference.tex
Given a set of polynomials $f_1, \dots, f_m \in \bC[x_1, \dots, x_n]$, a fundamental algorithmic task in computer algebra and algebraic geometry is the ``factorization'' of the ideal $I := (f_1, \dots, f_m)$ into ``irreducible'' components, generalizing the task of factoring a given polynomial into its irreducible factors.
Given the more complex structure of ideals in polynomial rings, the right generalization of the polynomial factoring problem is a primary decomposition of the ideal $I$, where one decomposes $I$ into ``primary ideals,'' which can be thought of, in the geometric sense, as the ``irreducible components'' of the zero set defined by the ideal $I$, accounted with their ``multiplicities.''
This task was first undertaken in the pioneering works of Lasker and Hermann~\cite{lasker1905theorie, hermann1926frage}, and due to its paramount importance in computer algebra and algebraic geometry, it has been the subject of extensive research by algebraic geometers, computer algebraists and complexity theorists ever since, as can be seen in \cite{seidenberg1974constructions,seidenberg1978constructions,gianni1988grobner,eisenbud1992direct,bayer1993can} and references therein.

The above task naturally leads to the basic problem of deciding whether a given ideal is \emph{prime}, which also has been extensively studied in theoretical computer science and algebraic geometry, as can be seen in \cite{heintz1981absolute,kaltofen1995effective,burgisser2004variations,burgisser2007differential,eisenbud2013commutative} and references therein.
Recall that an ideal $I \subset R$ is said to be prime if the quotient ring $R / I$ is an integral domain, equivalently if whenever $fg \in I$ we have $f \in I$ or $g \in I$. 
In the Turing model, which is the model that we will focus on in this work, the ideal primality testing problem can be formally defined in the following way.

\begin{problem}[Ideal primality testing]\label{problem: ideal primality testing}
    Given an algebraic circuit of size $s$ computing polynomials with integer coefficients $f_1, \dots, f_m \in \bZ[x_1, \dots, x_n]$, decide if the ideal $(f_1, \dots, f_m) \cdot \bC[x_1, \dots, x_n]$ prime.
\end{problem}

The ideal primality testing problem is the direct generalization of the problem of testing whether a single polynomial is (absolutely) irreducible.
More precisely, the absolute irreducibility testing problem for polynomials corresponds to \cref{problem: ideal primality testing} with $m = 1$.
A geometric version of the ideal primality testing problem was posed in \cite[Problem 8.5]{burgisser2004variations}: what is the complexity to check whether a given algebraic set is irreducible over $\bC$?
Algebraically, the latter problem can be cast as: given polynomials $f_1, \dots, f_m \in \bZ[x_1, \dots, x_n]$, is $\radideal{f_1, \dots, f_m}$ prime?

The $m=1$ case of \cref{problem: ideal primality testing}, that is, the absolute irreducibility test of polynomials, can be done in randomized polynomial time, and by several different methods, due to the works \cite{kaltofen1985fast,bajaj1993factoring,kaltofen1995effective,gao2003factoring} (and references therein). 
However, when $m \geq 2$, \cref{problem: ideal primality testing} seems to be notoriously more difficult, even for natural classes of ideals.
Currently, the only algorithms to test whether a (general) ideal is prime are algorithms which compute a primary decomposition of the input ideal~\cite{seidenberg1978constructions,gianni1988grobner,eisenbud1992direct} (and references therein).
Since these algorithms require the computation of \grobner{} bases as a subroutine, the best complexity class which upper bounds the ideal primality testing problem is $\EXPSPACE$, by applying the algorithm in \cite{gianni1988grobner} with the degree bounds for \grobner{} basis obtained in \cite{dube1990structure}.

Due to the challenges posed by the general version of \cref{problem: ideal primality testing} (and more generally for the primary decomposition problem), special cases of the problem(s) have been considered.
In particular, natural classes of ideals have been studied, such as the cases where the ideal is a \emph{complete intersection}\footnote{An ideal $I := (f_1, \dots, f_m)$ is a complete intersection when the codimension of $I$ equals the number of generators (in this case $m$). 
That is, each polynomial $f_i$ ``cuts'' the dimension of the algebraic set by one, just as linearly independent linear forms would cut the dimension of the space by one.
This is the algebraic generalization of the fact that a set of $m$ linearly independent forms defines a space of dimension $n-m$.}, or when the given ideal is \emph{radical}, or ideals with either constant dimension or constant codimension.
When the input to \cref{problem: ideal primality testing} is a complete intersection, the primary decomposition algorithm of \cite{eisenbud1992direct}, combined with the \grobner{} basis degree bounds of \cite{dickenstein1991membership} yields an exponential time algorithm.
When the input polynomials form a radical ideal, the works of \cite{chistov1986algorithm,grigoriev1986factorization} give an exponential time algorithm for computing the minimal primes of the given radical ideal, and \cite[Theorem 4.1]{burgisser2007differential} gives a $\PSPACE$-algorithm for counting the number of minimal primes of a given radical ideal.
Thus, for the class of radical ideals, \cref{problem: ideal primality testing} is in $\PSPACE$.\footnote{The works \cite{chistov1986algorithm,grigoriev1986factorization,burgisser2007differential} work on the geometric setting of the question, as stated in \cite[Problem 8.5]{burgisser2004variations}.
When the input ideal is already promised to be radical, then the geometric setting and \cref{problem: ideal primality testing} become the same problem.}
For the case of radicals of ideals generated by constantly many polynomials,  \cite[Theorem 1.1]{burgisser2010counting} gives an algorithm in $\mathrm{FRNC}$ for computing the number of minimal primes, therefore solving the primality testing problem in $\mathrm{RNC}$.
Lastly, the dimension-dependent \grobner{} basis degree bounds of \cite{mayr2017complexity}, when combined with the primary decomposition algorithm of \cite{gianni1988grobner}, yields a $\PSPACE$-algorithm for primary decomposition of ideals of constant dimension, thereby also providing a $\PSPACE$-algorithm for the ideal primality testing problem in this case.
Apart from the above described algorithms, an important approach to \cref{problem: ideal primality testing} is to provide sufficient conditions/criteria on certain classes of ideals that will ensure primality.
One prominent sufficient condition is given by the combination of Serre's criterion~(\cite[Theorem 11.5]{eisenbud2013commutative}, \cite[\href{https://stacks.math.columbia.edu/tag/031O}{Tag 031O}]{stacks-project}) with the Jacobian criterion~\cite[Theorem 16.19]{eisenbud2013commutative}, which works for the important class of connected Cohen-Macaulay ideals\footnote{The definition of Cohen-Macaulay ideals is somewhat technical, as can be seen in \cite[Chapter 18.2]{eisenbud2013commutative} and in \cite[\href{https://stacks.math.columbia.edu/tag/00N7}{Tag 00N7}]{stacks-project}.
However, as discussed in \cite[Chapter 18.5]{eisenbud2013commutative}, Cohen-Macaulay ideals form a rich class of ideals, and are central objects in commutative algebra.
In particular, zero dimensional ideals and complete intersections are Cohen-Macaulay.} (see \cite[Theorem 18.15]{eisenbud2013commutative}).

Given the seemingly weak upper bounds for the ideal primality testing problem even for the natural classes of ideals above, one can ask if there are any complexity-theoretic lower bounds for the ideal primality testing problem.
An easy argument shows that even for the natural class of zero-dimensional radical ideals, \cref{problem: ideal primality testing} is $\coNP$-hard.
We prove this result in \cref{proposition:conp hardness}.

In this work, assuming the Generalized Riemann Hypothesis (GRH), we substantially improve the complexity-theoretic gap between the upper and lower bounds for \cref{problem: ideal primality testing} for two important classes of ideals: radical ideals and equidimensional Cohen-Macaulay ideals.
Our main theorem is the following.

\begin{restatable}[Interactive protocols for primality]{theorem}{interactive}
\label{theorem: interactive protocols for primality}
    Let $C$ be an algebraic circuit of size $s$ with integer constants that computes $f_{1}, \dots, f_{m} \in \bZ\bs{x_{1}, \dots, x_{n}}$.
    If $I := \ideal{f_{1}, \dots, f_{m}}$ is either radical or equidimensional Cohen-Macaulay, and if the dimension of $I$ is given, then the complexity of testing if $I$ is prime lies in $\coAM$, assuming GRH.
\end{restatable}

Combining the above theorem with \cite[Theorem 4.1]{koiran1997randomized}, we obtain that \cref{problem: ideal primality testing} for radical ideals and for equidimensional Cohen-Macaulay ideals are in the third level of the Polynomial Hierarchy (PH).

\begin{restatable}[Primality testing in PH]{theorem}{primesPH}
\label{theorem: primality in PH}
    Let $C$ be an algebraic circuit of size $s$ with integer constants that computes $f_{1}, \dots, f_{m} \in \bZ\bs{x_{1}, \dots, x_{n}}$.
    If the ideal $I := \ideal{f_{1}, \dots, f_{m}}$ is either radical or equidimensional Cohen-Macaulay then the complexity of testing if $I$ is prime lies in $\Sigma_{3}^{p} \cap \Pi_{3}^{p}$, assuming GRH.
\end{restatable}

In the case when the given ideal is radical, \cref{theorem: primality in PH} makes substantial progress on  \cite[Problem 8.5]{burgisser2004variations}, and improves on the previous best upper bound of \cite{burgisser2007differential} from $\PSPACE$ to $\Sigma_3^p \cap \Pi_3^p$.
Note that the $\coNP$ lower bound from \cref{proposition:conp hardness} also applies to radical ideals.
Thus, our upper bound gets substantially close to the known lower bound for the problem.
Two important distinctions must be made here: our protocols only work when the input ideal is radical, and we need to assume the GRH; whereas the algorithm from \cite{burgisser2007differential} computes the number of irreducible components for the radical of the ideal generated by $f_1, \dots, f_m$, and their work is unconditional on the GRH.

In the case of equidimensional Cohen-Macaulay ideals, which contains as special cases the important classes of zero dimensional ideals and of complete intersection ideals, \cref{theorem: primality in PH} yields a nearly tight gap on the complexity of \cref{problem: ideal primality testing}, since the $\coNP$ lower bound of \cref{proposition:conp hardness} also applies to zero-dimensional radical ideals.
Note that our result for zero-dimensional ideals improves upon the previous upper bound of $\PSPACE$. 
If our input is a complete intersection, then the dimension is implicitly given to us, since the codimension of the ideal in this case is simply given by the number of input polynomials. 
Thus, in this case \cref{theorem: interactive protocols for primality} tells us that \cref{problem: ideal primality testing} is in $\coAM$.
This improves upon the previous best upper bound of exponential time.

The above results show that the most computationally expensive step in our protocols is to determine the dimension of the input ideal.
Thus, any improved protocol to compute exactly the dimension of a given ideal would lead to improvements to our primality testing protocols.

Testing the assumptions of \cref{theorem: interactive protocols for primality} is a hard computational problem.
Determining if an ideal is radical involves deciding if the ideal contains embedded primes.
As is lucidly explained in \cite[Section~4]{bayer1993can}, these embedded primes seem to be responsible for the worst case double exponential behaviour of many problems in commutative algebra and algebraic geometry.
The current best known algorithms for computing the radical of an ideal are essentially the same as the algorithms for computing primary decomposition \cite{gianni1988grobner, eisenbud1992direct}, and therefore take double exponential time in the worst case.
It is possible that the decision version of the problem is easier, however we are not aware of any algorithms that solve the decision version faster than in double exponential space.
Deciding if an ideal is equidimensional Cohen-Macaulay is similarly difficult, current methods rely on testing non-vanishing of Ext modules, which in turn requires the computation of projective resolutions.
The complexity of this task is double exponential.
This discussion also sheds some light on what makes the primality testing problem tractable in the cases we consider: we assume that our ideals do not have any embedded primes.

As mentioned above, previous works solved the primality testing problem by either computing a primary decomposition of the input ideal, or by using de Rham cohomology to count the number of irreducible components of the associated algebraic set (for the case of radical ideals).
Our approach to solve \cref{problem: ideal primality testing} is different from the previous approaches. 
Our strategy is inspired by Koiran's approach~\cite{K96} to decide whether a system of polynomial equations has a solution, which we now describe.

Koiran, in his seminal work \cite{K96}, proposed a fundamentally different approach to the problem of deciding whether a system of polynomial equations with integer coefficients has a solution over $\bC$ (known as the Hilbert's nullstellensatz problem).
His approach to this problem was based on the following idea: let $\cF := \{f_1 = 0 , \dots, f_m = 0 \}$, where $f_i \in \bZ[x_1, \dots, x_n]$ with $\deg f_i \leq d$, be our system of polynomial equations.
By Hilbert's nullstellensatz, $\cF$ is unsatisfiable if, and only if, $1 \in \ideal{f_1, \dots, f_m}$. 
In other words, $\cF$ is unsatisfiable iff there are polynomials $g_1, \dots, g_m \in \bC[x_1, \dots, x_n]$ such that $1 = f_1 g_1 + \cdots + f_m g_m$.
Note that this last equation is simply a linear system of equations with integer entries, where the variables are the coefficients of the $g_i$'s. 
Thus if there is a solution for this system over $\bC$, there must be a solution over $\bQ$.
Hence, by clearing denominators, there are polynomials $h_1, \dots, h_m \in \bZ[x_1, \dots, x_n]$ and $a \in \bZ_{>0}$ such that $a = f_1 h_1 + \cdots + f_m h_m$.
In particular, if $\cF$ is unsatisfiable, given any prime $p \in \bN$ which does not divide $a$, we have $1 \equiv a^{-1} \cdot (f_1 \cdot h_1 + \cdots + f_m h_m) \bmod p$, which implies that $\cF$, when seen as a system over $\bF_p[x_1, \dots, x_n]$, does not have any solutions over the field $\overline{\bF}_{p}$.
In particular, $\cF$ has no solutions over $\bF_p^n$.
The main conceptual insight in \cite{K96} is to ask (and answer) the question: if $\cF$ is satisfiable, would we be able to find "simple" solutions (that is, solutions in $\bF_p^n$, as opposed to solutions in $\overline{\bF}_{p}^n$) for \emph{enough} primes $p$?\footnote{We would like to emphasize that the part which takes a lot of work in this claim is the part that says that satisfiable systems will be satisfiable in $\bF_p^n$. 
Checking that a system remains satisfiable over $\overline{\bF}_p^n$ is also easy, as it can be solved by a similar argument as the one given for the unsatisfiability case.
}
His main theorem \cite[Theorem 1]{K96} proves (assuming the GRH) a \emph{quantitative} version of the following fact: if $\cF$ is unsatisfiable, then there are "very few" primes $p$ such that $\cF$ has a solution over $\bF_p^n$; if $\cF$ is satisfiable, then there are "enough" (small enough) primes $p$ such that $\cF$ has a solution over $\bF_p^n$.  
Equipped with the above theorem, Koiran can now distinguish between the two cases by using a set lower bound protocol, hence putting the Hilbert's nullstellensatz problem in $\coAM$.

It is important to emphasize that the main technical challenge with the above approach, as pointed out by Koiran (see \cite[Page 275]{K96}), is to prove such quantitative bounds on the number of "good primes," when one changes the base ring from $\bZ$ to $\bF_p$ (or geometrically, from $\bC$ to $\overline{\bF}_{p}$).
The principle of changing the base field where an object of interest lives (in our case our algebraic set over $\bC^n$) to extract some information of our geometric object based on properties of its image is called a \emph{base change theorem}.
Such results are well studied in algebraic geometry and number theory, as can be seen in \cite{grothendieck1965elements} and \cite[Appendix C]{poonen2008rational}.\footnote{Sometimes base change theorems are also called \emph{spreading out theorems}, as is done in the references given here.}
Often times, such base change theorems are non-constructive, which suffices for their mathematical purposes.
However, for such a base change theorem to be useful for computation, as is the case in \cite{K96}, one needs to prove  \emph{effective versions} of such base change theorems. 
As usual in mathematics and computer science, making such theorems effective often requires a considerable amount of work.

In this work, we use the same high-level strategy developed by Koiran to give an $\AM$ protocol for proving that a given input ideal (from one of our classes of ideals) is not prime, assuming that we are given the dimension of our input ideal. 
We begin by noticing that in the classes of ideals that we study, there are three ways in which our ideal can fail to be prime: the ideal can have two minimal primes of different dimensions, or all the minimal primes of the ideal have the same dimension but there are at least two of them, or the ideal can have a single minimal prime with "multiplicity." 
For radical ideals, only the first and the second cases can happen, whereas for equidimensional Cohen-Macaulay ideals only the second and third cases can happen (\cite[Corollary~18.14]{eisenbud2013commutative}).

In the case where we have two minimal primes with different dimensions (which can only happen in the radical case) we can identify that the given ideal is not prime via the Jacobian.
This allows us to reduce this case to an instance of the Hilbert's nullstellensatz problem, which can be handled by the protocol of \cite{K96}.
This is done in \cref{subsection: interactive  radical}.

In the case where we have a single minimal prime with multiplicity, (which can only happen in the equidimensional Cohen-Macaulay case), we can use Serre's criterion, combined with the Jacobian criterion (as done in \cite[Theorem 18.15]{eisenbud2013commutative}) to identify the multiplicity via the codimension of the Jacobian ideal inside our ideal.
Thus, in this case we can use the protocol to compute dimension lower bounds given by \cite{koiran1997randomized}.
This is done in \cref{subsection: interactive CM}.

Now, the only case we have left is when our ideal is equidimensional, but has more than one irreducible component (this can happen for both the radical and the equidimensional Cohen-Macaulay cases).
This turns out to be the most challenging case for us.
The base change theorems of \cite{grothendieck1965elements} tell us that if our ideal is prime over $\bC[x_1, \dots, x_n]$, then, for many choices of primes $p$, our ideal over $\bF_p[x_1, \dots, x_n]$ will remain prime over $\overline{\bF}_p[x_1, \dots, x_n]$.
Combining this fact with \cite{seidenberg1974constructions}, one also obtains that for many choices of primes $p$, if our ideal is equidimensional and not prime, then our ideal over $\bF_p[x_1, \dots, x_n]$ will remain equidimensional and not prime over $\overline{\bF}_p[x_1, \dots, x_n]$.
Since we are only interested in distinguishing whether our ideal has more than one component of top dimension, by an effective version of Bertini's theorem (\cref{cor:effectivebertini}), we can assume that we are working with ideals of dimension 1.
Thus, we have reduced our main problem to the task of distinguishing whether an algebraic curve (of potentially exponential degree) is absolutely irreducible or not.

While the above gives us hope to be able to distinguish this case by going modulo $p$ (for a "good prime" $p$), there are two problems that we need to address: can we make this base change theorem \emph{effective}? 
And in case we can make it effective, how can we witness the distinction between an absolutely irreducible curve (i.e. a prime ideal) and a reducible curve (i.e., non-prime ideal) in a "simple way"?
Could we do the latter simply by looking at points in $\bF_p^n$, as is done by Koiran?
As it turns out, the answer to the last question is affirmative, due to effective versions of the celebrated Lang-Weil theorem (\cref{thm:effective lang weil}). 

We have thus far reduced our primality testing problem to the problem of proving an effective base change theorem which preserves equidimensional components of given algebraic sets.
However, some care needs to be taken here: the Lang-Weil theorem only guarantees enough solutions over $\bF_p^n$ of algebraic varieties which are \emph{defined} over $\bF_p$ (that is, the polynomials defining the variety must have coefficients in $\bF_p$).
Thus, if we are to use the Lang-Weil theorem, we must ensure that for an equidimensional non-prime ideal, \emph{at least two} of its minimal primes \emph{will be $\bF_p$-definable} for \emph{enough} "good primes."
At this point, one may ask:\footnote{As mentioned in \cite[pages 47 and 48]{bayer1993can}, "one who has never done any calculations."} why not just preserve every minimal prime? 
As discussed in \cite[Pages 47 and 48]{bayer1993can}, this is a very bad idea, since to do so would cause us to have to work over polynomials defined over extremely large extension fields, and therefore kill any attempt to get an effective base change theorem, since the coefficients become "too complex."

With the above in mind, we have now finally arrived at the precise effective base change theorem that we need to prove: given an equidimensional ideal $I := (f_1, \dots, f_m)$ with $f_i \in \bZ[x_1, \dots, x_n]$ and $\dim I = 1$, we need to ensure that there are enough "good primes" such that \emph{at least two} of its irreducible components are defined over $\bF_p$.
To prove such an effective base change theorem, we combine Kaltofen's factoring algorithm over the algebraic closure~\cite{kaltofen1995effective} with results from algebraic number theory and elimination theory.
As we have mentioned before, this is the crucial technical result that we need, and it is carried out in \cref{section: height bounds,section: base change}, which culminate in our main technical theorems, \cref{thm:irredpreserve,thm:redpreserve}.

We now present a summary of the above high-level discussion of our results, and in the following subsection we present a more complete overview of our proof, where we discuss in a more precise form all the issues that need to be overcome to prove correctness of our interactive protocols.

\input{tldr}

\subsection{Proof overview}

We now give a more detailed overview of the proof of our main technical theorem.
Recall that we are given polynomials $f_1, \dots, f_m \in \bZ[x_1, \ldots, x_n]$, defining the ideal $I := \ideal{f_{1}, \dots, f_{m}}$ over $\bC[x_1, \dots, x_n]$.
Let $V := Z(I) \subseteq \bC^n$ be the zero set of the ideal $I$, and $r:= \dim V$ be the dimension of the ideal.
Since we will be interested in the bit-complexity of the polynomials and (algebraic) integers involved, we denote the \emph{logarithmic height} of an integer $a \in \bZ$ by $\lheight{a} := \lceil \log(|a| + 1) \rceil$ (i.e., its bit complexity) and extend this definition to denote the logarithmic height of a polynomial by the maximum logarithmic height of any of its coefficients.

As stated earlier, for the types of ideals we consider, there are only a few ways in which the ideal $I$ can fail to be prime.
If $I$ is radical, then it can fail to be prime because $V$ consists of at least two components, either both of dimension $r$, or one of dimension $r$ and one of dimension less than $r$.
If $I$ is equidimensional CM, it can fail to be prime either because $V$ consists of two components of dimension $r$, or because $V$ consists of a single component of dimension $r$, but this component occurs with "multiplicity", that is, the unique primary component corresponding to the minimal prime of $I$ is not itself prime.

Handling the case where $V$ has a component of dimension less than $r$ (which can happen only in the radical case), or the case where $V$ has a single component of dimension $r$ but with multiplicity (which can happen only in the CM case) is straightforward, as we have discussed in the previous section.
Thus, we focus on the case when $I$ is equidimensional, where deciding primality reduces to deciding if $V$ has only one component of dimension $r$.
Also, as mentioned before, by an effective version of Bertini's theorem, we can further assume that $r = 1$.

As discussed above, our approach involves studying how the ideal $I$ and the zeroset $V$ behave when the coefficients are changed from $\bZ$ to $\bF_{p}$ for some prime $p$.
Let $I_{p}$ be the ideal generated by $f_{1} \pmod p, \cdots, f_{m} \pmod p$, and $V_{p} := Z(I_{p}) \subseteq \overline{\bF}_{p}^{n}$.
We need to show that three properties are preserved.
The first is that for all but a small number of primes, $\dim V = \dim V_{p}$.
The second is that if $V$ is irreducible, then for all but a small number of primes, $V_{p}$ is irreducible.
The last is that if $V$ is reducible, then for sufficiently many primes, $V_{p}$ is also reducible, and crucially, $V_{p}$ has at least two distinct components that are $\bF_{p}$-definable.

Before we sketch our proof techniques for the above facts, we show that we cannot hope to get much stronger statements than the ones we claim.
While the examples we will give here may appear simple, it is easy to generalize their main idea to produce examples with more complex behavior.

For the first property: dimension. 
Consider for example $I = \ideal{x_{1}, x_{1} + ax_{2}}$ for some $a \in \bN$.
We have $\dim V = n-2$. 
However, for any prime $p \mid a$, we have $I_{p} = \ideal{x_{1}}$ therefore $\dim V_{p} = n-1$.
Thus, we cannot hope that $\dim V_p$ remains unchanged for every base change from $\bZ$ to $\bF_p$.

For the second property: irreducibility.
In this case, let $I = \ideal{x_{1} (x_{1} - 1) - a x_{2}}$. 
Here, we have that $V$ is irreducible, but for any prime $p | a$, the zeroset $V_{p}$ is reducible.

And now for the third property: reducibility.
Let's assume $n = 2$.
Let $g \in \bZ\bs{x_{1}}$ be a monic polynomial of degree $d$, with Galois group $S_{d}$ (the full symmetric group on $d$ elements).
Note that a random polynomial satisfies this property.
Let $f$ be the polynomial obtained by homogenizing $g$ with respect to $x_{2}$, and let $I = \ideal{f}$.
Then $\dim V = 1$, and $V$ has $d$ irreducible factors.
The smallest extension of $\bQ$ that contains all the factors of $f$ is exactly the splitting field of $g$, which has degree $d!$, which is exponential in $d$.
However, there are extensions of degree just $d^{2}$ that contain two irreducible factors, namely extensions generated by any pair of roots of $g$.

Now that we are aware of the pitfalls along the way, we can give a more precise idea of the quantitative bounds that we obtain.
As is evident by the examples, the size of the set of bad primes for each of the statements will depend on the logarithmic height 
of the given polynomials.
We will show that the size of the set of bad primes for the first two properties is polynomial in the degrees and logarithmic heights of $f_{1}, \dots, f_{m}$, and exponential in the number of variables.
As the logarithmic heights are themselves exponential in $s$ (the size of the circuit), the bounds we obtain are exponential in the input size.

With the above in mind, in order for the set size lower bound protocols to work, when $V$ is reducible, the number of good primes (that is, the number of primes $p$ for which there are at least two $\bF_{p}$ definable components) has to also be exponential, within a bound $A$ of primes up to which we can check roots.
In \cref{remark: reducible base change tight} we show that the best one could hope to do in terms of the density of good primes is inverse exponential.
Our main theorem will show that we can indeed achieve this inverse exponential density.
This will show that there are sufficiently many good primes.

There are three main ideas/tools we use to prove the above properties.
The first is the fact that ideal membership can be written as a linear system, if bounds are known for the degrees of the witnesses.
Here, by a witness to the membership of $g \in I$ we mean polynomials $h_{1}, \dots, h_{m}$ such that $g = \sum f_{i} h_{i}$.
In general these bounds can be doubly exponential on the input size \cite{hermann1926frage}. 
However in the special cases of radical ideals and ideals of constant dimension, single exponential bounds are known (\cite{jelonek2005effective}, \cite{mayr2013dimension}).
This allows us to show that certain memberships (and non-memberships) in ideals continue to hold when going modulo $p$, for all but a small number of primes.
The second tool is an effective Bertini-Noether theorem \cite{kaltofen1995effective}.
The Bertini-Noether theorem states that if a polynomial $f \in \bZ\bs{x_{1}, \dots, x_{n}}$ is absolutely irreducible, then for all but finitely many primes $p$ it remains absolutely irreducible in $\bF_{p}$.
The effective version of this theorem states gives an upper bound on the number of bad primes for  $f$ that depends on the logarithmic height of $f$.
The third main tool we use is the factoring algorithm of \cite{kaltofen1995effective}.
This algorithm gives a bound on the logarithmic heights of the coefficients of factors of polynomials.
In this work, we combine the above ideas with a number of standard tools from commutative algebra and algebraic number theory, which can be found in \cref{section: preliminaries,section: height bounds}.

We now return to the three statements we want to prove.

\paragraph{Dimension of $V_{p}$.}
We want to show that $\dim V_{p} = 1$ for all but a small number of primes $p$ (recall that we are assuming $\dim V = 1)$.
Since $\dim V = 1$, we have $I \cap \bZ\bs{x_{i}, x_{j}} \neq \ideal{0}$ for all pairs $i \neq j$.
This is because $I \cap \bZ\bs{x_{i}, x_{j}}$ corresponds to projection to coordinates $i, j$, and such projections also have dimension at least $1$.
Further, if $\dim V = 1$, then $I \cap \bZ\bs{x_{i}} = (0)$ for some $i$, say $i = 1$.
This is because projection to every coordinate cannot be finite, since $V$ is not finite.
These two properties characterise the fact that $\dim V = 1$.

By the effective Nullstellensatz, if $I \cap \bZ\bs{x_{i}, x_{j}} \neq \ideal{0}$, then it contains a polynomial $g_{ij}$ with $\deg{g_{ij}} \leq d^{n}$, and such that $g_{ij} = \sum f_{k} h_{ijk}$ with $h_{ijk} \in \bQ[x_1, \dots, x_n]$ having $\deg{h_{ijk}} \leq d^{n}$. 
Moreover, by the effective nullstellensatz, a similar linear-algebraic condition can be derived (with the same degree bounds) which established that there is no non-zero polynomial $g_1 \in I \cap \bZ\bs{x_{1}}$.
Since the degrees of the potential $h_{ijk}$ are bounded, we can write the condition for existence of $g_{ij}, g_{1}$ as a linear system, where the unknowns are the coefficients of $h_{ijk}$.
Using standard height bounds, we can deduce that for all but exponentially many primes, the existence and non-existence of solutions is preserved mod $p$, since these facts only depend on the vanishing and non vanishing of certain minors of the matrix of this linear system.
Therefore, for all but exponentially many primes we have $\dim V_{p} = 1$.
This result is formally stated and proved in \cref{lem:dimpreserve}.

\paragraph{Irreducibility of $V_{p}$.}
We want to show that if $V$ is irreducible, then $V_{p}$ is irreducible for all but a small number of primes $p$.
Suppose $n = 2$, and $\dim V = 1$.
Then $V$ is irreducible if and only if $g \in \radideal{I}$ for some absolutely irreducible polynomial $g$.
Further, any such $g$ must divide every generator of $I$.
This easily allows us to get a bound on the coefficients of $g$, and by the effective Nullstellensatz, we can deduce that $g^{e} = \sum f_{i} h_{i}$ with $\deg{g^{e}}, \deg{h_{i}} \leq d^{n}$.
By Cramer's rule, we can upper bound the common denominator of the polynomials $h_{i}$.
For any prime $p$ that does not divide the common denominator, we have $g^{e} \in I_{p}$.
Further, by an effective Bertini-Noether theorem, for all but a few primes $p$, the polynomial $g \pmod p$ remains absolutely irreducible.
Finally, for all but a few primes, $\dim V_{p} = \dim V$.
If $p$ is any prime that passes all the above conditions, then $V_{p}$ is irreducible.

The general case, where $n$ is arbitrary, is slightly more involved.
It is no longer true that $V$ is irreducible if and only if $g \in \radideal{I}$ for some absolutely irreducible polynomial $g$.
For example, $V$ could consist of two curves that both lie on the same irreducible hypersurface.
To overcome this, we project to a random two dimensional linear subspace, call this projection $\pi$.
Since $V$ is irreducible, $\pi(V)$ is irreducible, and $\pi(V) = Z(g)$ for an absolutely irreducible polynomial $g$.
A technical (and somewhat involved) argument gives us bounds on the coefficients of $g$ in this more general setting (\cref{lem:eliminationextendedeffective}).
Proceeding as above, we can show that for all but a few primes $p$, we have $\dim V_{p} = 1$ and $\pi(V_{p})$ is irreducible.

However, the above is not enough to conclude that $V_{p}$ is irreducible, since $V_{p}$ could have been reducible, and all of its components could have been mapped to $\pi(V_{p})$ under $\pi$.
To fix this, instead of just considering a single random projection $\pi$, we consider a large number of random projections $\pi_{1}, \dots, \pi_{k}$, and repeat the above argument.
If $V_{p}$ has two components, $\pi_{i}(V_{p})$ can only be irreducible for a small fraction of these projections, and we can use this to bound all primes $p$ such that $V_{p}$ is not irreducible.
This result is formally stated and proved in \cref{thm:irredpreserve}.

\paragraph{Reducibility of $V_{p}$ and $\bF_{p}$-definability of some of its components.}
We want to show that if $V$ is reducible, then $V_{p}$ is reducible, and has at least two $\bF_{p}$-definable components for sufficiently many primes $p$.
Again, let us start by assuming $n = 2$.
If $V$ has $k$ irreducible components, then $V = Z(g)$ for a polynomial $g \in \radideal{I}$ with exactly $k$ absolutely irreducible factors.
As before, $g$ must divide the generators of $I$, which allows us to get a bound on the coefficients of $g$.
By the effective Nullstellensatz, we can deduce that $g^{e} = \sum f_{i} h_{i}$ with $\deg{g^{e}}, \deg{h_{i}} \leq d^{n}$, and we can deduce that for all but a few primes $p$, we have $g^{e} \in I_{p}$.
However, this does not imply that $V_{p}$ is reducible: it might be the case that $V_{p} = Z(g')$ for some $g'$ that is an absolutely irreducible factor of $g \pmod p$.
Therefore, the strategy we used to show that irreducibility is preserved does not work here.

To prove that the above cannot happen (that is, that $V_p = Z(g')$), we try and find defining equations for some components of $V$.
In the case when $n = 2$, these are exactly the factors of $g$.
Suppose $g = \prod_{i=1}^{k} g_{i}$.
Each $g_{i}$ has coefficients in some algebraic extension of $\bQ$.
The smallest extension of $\bQ$ that contains every $g_{i}$ might be of prohibitively high degree, but from the factoring algorithm of Kaltofen (\cref{lem:kaltofen factoring algorithm}) we can deduce that there is a small enough extension of $\bQ$ that contains $g_{1}, g_{2}$.
Now we consider the ideals $I_{1} := I_{1} + g_{1}$ and $I_{2} := I + g_{2}$.
These are ideals of dimension $1$ that are irreducible.
Therefore, we can use our base change and irreducibility arguments to deduce that  $I_{1, p}$ and $I_{2, p}$ are irreducible curves, which will be components of $I_p$.
Of course we have to be careful here: $g_{1}$ and $g_2$ have coefficients in $\bQ\br{\alpha}$ for some algebraic number $\alpha$, therefore it is not even clear at first what it means to go modulo $p$.
Here, we need to invoke some standard results from algebraic number theory.
If $G$ is the monic equation of $\alpha$, then for any $p$ such that $G \pmod p$ has a root, the operation of going mod $p$ is well defined on $\bZ\bs{\alpha}\bs{x_{1}, \dots, x_{n}}$.
This allows us to not only make sense of the mod $p$ operation and invoke our previous result, but it also guarantees that the components of $V_{p}$ that we find this way are $\bF_{p}$ definable.
It is important to note here that there are a number of technical challenges in extending our earlier statements to this general algebraic number theoretic setting. 
These challenges are dealt with in \cref{section: height bounds}.

Now suppose $n$ is arbitrary.
We again find equations that define the components of $V$.
As before we do this by random projection to a two dimensional affine subspace, and then we apply (some of) our arguments from the case of $n = 2$.
It turns out that a single extra polynomial $g_{1}$ is not enough to define a component of $V$. 
However, we show any component of $V$ can be defined by adding at most two equations to $I$.
As before, we give bounds on the coefficients of these equations, and then invoke our earlier arguments to deduce that $V_{p}$ is reducible, and for sufficiently many primes it has at least two components that are $\bF_{p}$ definable.
This result is formally stated and proved in \cref{thm:redpreserve}.

\subsubsection{Putting it all together}
We now show how all the results mentioned above are put together to obtain our main protocol.
For brevity, we only consider the case when $I$ is radical and $\dim V = 1$, as the other cases are very similar.
We now sketch an $\AM$ protocol that shows that $I$ is not prime.

In the radical case, $I$ can fail to be prime for 2 reasons: $V$ may have a component of strictly smaller dimension (hence in this case $V$ will have at least one isolated point as a solution), or $V$ is equidimensional with at least two components.

\paragraph{Case 1: $V$ has an isolated point as a solution.}
Consider the Jacobian matrix $\cJ$, whose entries are $\cJ_{ij} = \partial_{i} f_{j}$.
If $\alpha$ is an isolated point of $V$, then the rank of $\cJ(\alpha)$ is $n$.
Merlin first sends Arthur the index of $n$ columns of $\cJ$ that are linearly independent.
If this submatrix is $\cJ'$, then Arthur simply has to check if the system $f_{1} = 0, \dots, f_{m} = 0, \det \cJ' \neq 0$ is satisfiable.
This can be rewritten as $f_{1} = 0, \dots, f_{m} = 0, y \det \cJ' -1 = 0$ where $y$ is a new variable.
Now, the satisfiability protocol of \cite{K96} can be used to decide $I$ is not prime, as in this case the protocol of \cite{K96} will return that the system that we constructed will have a solution (which must correspond to an isolated point).

On the other hand, if $V$ only has components of dimension $1$, then at every point on $V$, the rank of $\cJ$ is at most $n-1$.
Therefore, no matter what choice of columns Merlin picks, the above protocol will fail to accept with high probability.

\paragraph{Case 2: $V$ has two components of dimension $1$.}
For every prime $p$ such that $V_{p}$ has at least two irreducible components, the Lang-Weil theorem says that the number of $\bF_{p}$ points in $V_{p}$ is at least $2p$ (up to an error term).
If $V$ has two components, then we have shown that there are enough primes with this property.
On the other hand, if $V$ is irreducible, for all except a finite number of primes, $V_{p}$ is irreducible, and therefore the Lang-Weil theorem says that the number of $\bF_{p}$ points in $V_{p}$ is at most $p$ (up to an error term).
Therefore, Merlin can prove to Arthur that $V$ is reducible, by performing the set lower bound protocol on the set of primes for which the number of $\bF_{p}$ points in $V_{p}$ is at least $2p$.
Merlin also has to convince Arthur that there are in fact $2p$ points in $V_{p}$, since the primes $p$ are too large for Arthur to check this himself.
For this, the set lower bound protocol is used again, on the set of $\bF_{p}$ points of $V_{p}$.

%% file: tldr.tex
\paragraph{Summary of contributions.}

\begin{enumerate}
    \item In this work, we give $\AM$ protocols for non-primality for natural classes of ideals, assuming the Generalized Riemann Hypothesis. 
    This significantly tightens the complexity gap of these problems.
    More precisely, we have:
    \begin{itemize}
        \item For the class of radical ideals, we prove that \cref{problem: ideal primality testing} can be solved in $\coAM$, whereas \cref{proposition:conp hardness} shows that \cref{problem: ideal primality testing} is $\coNP$-hard for this class of problems.
        \item For the class of equidimensional Cohen-Macaulay ideals (which includes the class of dimension zero ideals, and the class of complete intersection ideals), we prove that \cref{problem: ideal primality testing} can be solved in $\coAM$, whereas \cref{proposition:conp hardness} shows that \cref{problem: ideal primality testing} is also $\coNP$-hard for this class of ideals.
    \end{itemize}
    \item On the technical side, our main contribution is to prove effective base change theorems for geometric irreducibility (and reducibility) of algebraic sets with equations defined over the integers.
    This is the content of \cref{thm:irredpreserve,thm:redpreserve}.
    \begin{itemize}
        \item One important remark is that our base change theorems, together with the effective Chebotarev density theorem, yields nearly optimal density bounds for the good primes preserving (ir)reducibility.
        This result is interesting on its own right.
        \item The proof of our base change theorem involves the combination of (effective) techniques from algebraic geometry, elimination theory, efficient algorithms for factorization over the algebraic closure of a field (and its corresponding absolute irreducibility test), as well as techniques from algebraic number theory (to control the bit complexity of the algebraic numbers involved in the operations that we need to consider).
    \end{itemize}
\end{enumerate}

%% file: preliminaries.tex
In this section we establish the notation that will be used throughout the paper, along with the background results that we will need in the later sections.

\subsection{Notation} \label{subsec:notation}
Let $R := \bZ\bs{x_{1}, \dots, x_{n}}$ denote the $n$-variate polynomial ring with integer coefficients, and let $S := \overline{\bQ}\bs{x_{1}, \dots, x_{n}}$ denote the extension of scalars to the algebraic closure of $\bQ$.
All the ideals we discuss are ideals of $S$, unless stated otherwise.
The same holds for any property of the ideals we discuss, for example the property of an ideal being prime, or Cohen-Macaulay, among others.
We use $\bA^{n}$ to denote the affine space in $n$-dimensions, with underlying field $\overline{\bQ}$.
We use $Z(I)$ to denote the zeroset of an ideal $I$.

In the course of our proofs, we will have to deal with polynomials with coefficients in finite extensions of $\bQ$.
Given a monic irreducible polynomial $q \in \bZ\bs{z}$, and a root $\alpha \in \overline{\bQ}$ of $q$,  we use $A$ to denote the ring $\bZ\bs{\alpha}\bs{x_{1}, \dots, x_{n}}$.
Note that $\bZ\bs{\alpha}$ is not necessarily the ring of integers of $\bQ\br{\alpha}$.
Elements of $A$ will be represented as polynomials in $R\bs{z}$, with $z$-degree less than $\deg q$.
This representation of elements of $A$ is unique.
The ring $A$ depends on $q$, however $q$ will always be clear from context, and therefore we suppress it from notation.
In this setting, by $\deg{f}$ we mean the degree of $f$ in the $x$-variables, unless stated otherwise.

The logarithmic height of an integer $c$, denoted by $\lheight{c}$, is the bit-complexity of $c$.
This notion of logarithmic heights extends naturally to polynomials.
Given a polynomial $f \in R$, the logarithmic height of $f$, also denoted by $\lheight{f}$, is the maximum logarithmic height of the coefficients of $f$.
If $f \in A$, then $\lheight{f}$ is defined to be the logarithmic height of $f$ treated when written as a polynomial in $R\bs{z}$ with $z$-degree less than $\deg q$.

Given polynomials $f_{1}, \dots, f_{m}, g$ with $\deg{f_{i}},\deg{g} \leq d$, the condition $g \in \ideal{f_{1}, \dots, f_{m}}$ is equivalent to the existence of polynomials $h_{1}, \dots, h_{m}$ such that $g = \sum f_{i} h_{i}$.
If a bound on the degrees of $h_{i}$ is known a priori (which is usually the case), then the above condition can be written as a linear system where the coefficients of $h_{1}, \dots, h_{m}$ are the unknowns.
If this bound on the degrees of $h_{i}$ is $D$, we use $M_{D}(f_{1}, \dots, f_{m})$ to denote the matrix corresponding to this linear system.
The entries of this matrix are coefficients of $f_{1}, \dots, f_{m}$.
When $f_{1}, \dots, f_{m}$ are clear from context we denote this matrix by just $M_{D}$.
Observe that the total number of columns of $M_{D}$, that is, the number of unknowns in the system is at most $m \cdot \binom{D+n}{n}$.
The total number of equations is at most $\binom{D + d + n}{n}$.
Both these estimates easily follow from counting the number of monomials of given degrees.
When reasoning using this linear system, we will slightly overload notation, and use the same symbols to refer to both polynomials and their coefficient vectors.
For example, we write the condition $g \in \ideal{f_{1}, \dots, f_{m}}$ as $g = M_{D} v$ where $v$ is a vector of unknowns.
We further say that $h_{1}, \dots, h_{m}$ are a solution to the system.

\subsection{Results from complexity theory}

We begin by describing an AM protocol that lower bounds the size of sets that have an efficient membership test.
The protocol is due to \cite{goldwasser1986private}, and the following statement is from \cite[Section~8.4.1]{arora2009computational}.

\begin{lemma}
    \label{lem:goldwassersipser}
    Suppose $S \subset \bc{0, 1}^{n}$ is a set such that the problem of membership in $S$ is in NP.
    Suppose further that a number $K$ is known, and $S$ is guaranteed to either satisfy $\abs{S} \geq 2K$ or $\abs{S} < K$.
    Then the problem of deciding if $\abs{S} \geq 2K$ is in $\AM$.
\end{lemma}
We sketch the protocol here, since it will make it easier to explain the generalisation we require.
However, we omit proofs.
We also assume for convenience that $K$ is a power of $2$.
\begin{algorithm}
  \caption{Goldwasser-Sipser set lower bound protocol
    \label{alg:goldwasser sipser}}
    \SetKwInOut{Input}{Input}
    \SetKwInOut{Output}{Output}
    \SetKwInOut{Merlin}{Merlin}
    \SetKwInOut{Arthur}{Arthur}

    \Input{Boolean formula $\phi_{S}(x, y)$ such that $x \in S$ if and only if there exists $y$ satisfying $\phi_{S}(x, y)$, and an integer $K$.
    }

    \Arthur{Let $k := \log_{2}(2K)$, and pick a random hash function $h:\bc{0, 1}^{n} \to \bc{0, 1}^{k+1}$ from a pairwise independent hash function collection, and pick $u \in \bc{0, 1}^{k+1}$ uniformly at random.
    Send $h, u$ to Merlin.}

    \Merlin{Find $x, y$ such that $h(x) = u$ and $\phi_{S}(x, y)$ is true.
    Send $x, y$ to Arthur.}

    \Arthur{Accept if and only if $h(x) = u$ and $\phi_{S}(x, y)$ is true.}
    
\end{algorithm}

We now state a slight generalisation of this protocol, to the case when membership in $S$ itself can be verified by an AM protocol.

\begin{corollary}
    \label{cor:goldwasser sipser extended}
    Suppose for each $x \in \bc{0, 1}^{n}$ there exists a subset $S_{x} \subset \bc{0, 1}^{m(x)}$, and an integer $K(x)$, such that $m(x), K(x)$ are polynomially bounded functions of $x$.
    Suppose further that there is a uniform algorithm that runs in polynomial time that given input $x$, returns the number $K(x)$ and also returns a boolean formula $\phi_{x}$, such that $z \in S_{x}$ if and only if there exists $y$ such that $\phi_{x}(z, y)$ is true.
    Suppose further than an integer $K$ is known.
    
    If $S \subset \bc{0, 1}^{n}$ is the set of elements $x$ such that $\abs{S_{x}} \geq 2K(x)$, and if $S$ is promised to either satisfy $\abs{S} \geq 2K$ or $\abs{S} \leq K$, then the problem of deciding if $\abs{S} \geq 2K$ is in $\AM$.
\end{corollary}

We give a four round protocol that decides $\abs{S} \geq 2K$.
The result then follows from the fact that $\AM = \AM[c]$ for all constants $c \geq 2$.
The proof of correctness of the protocol is the exact same as the proof of correctness of \cref{lem:goldwassersipser}, therefore we omit the proof.

\begin{algorithm}
  \caption{Goldwasser-Sipser set lower bound protocol with membership in AM
    \label{alg:goldwasser sipser extended}}
    \SetKwInOut{Input}{Input}
    \SetKwInOut{Output}{Output}
    \SetKwInOut{Merlin}{Merlin}
    \SetKwInOut{Arthur}{Arthur}

    \Input{An integer $K$, and an algorithm that on input $x$ outputs the integers $K(x)$ and the circuit $\phi_{x}$ described in \cref{cor:goldwasser sipser extended}
    }

    \Arthur{Let $k := \log_{2}(2K)$, and pick a random hash function $h:\bc{0, 1}^{n} \to \bc{0, 1}^{k+1}$ from a pairwise independent hash function collection, and pick $u \in \bc{0, 1}^{k+1}$ uniformly at random.
    Send $h, u$ to Merlin.}

    \Merlin{Find $x$ such that $h(x) = u$ and $\abs{S_{x}} \geq 2K(x)$.
    Send $x$ to Arthur.}

    \Arthur{Reject if $h(x) \neq u$.
    If $h(x) = u$, let $k_{x} := \log_{2}(2K(x))$, and pick a random hash function $h_{x}:\bc{0, 1}^{m(x)} \to \bc{0, 1}^{k_{x}+1}$ from a pairwise independent hash function collection, and pick $u_{x} \in \bc{0, 1}^{k_{x}+1}$ uniformly at random.
    Send $h_{x}, u_{x}$ to Merlin.}

    \Merlin{Find $y, z$ such that $h_{x}(z) = u_{x}$ and such that $\phi_{x}(z, y)$ is true.
    Send $z, y$ to Arthur.}

    \Arthur{Accept if and only if $h_{x}(z) = u_{x}$ and $\phi_{x}(z, y)$ is true.}
    
\end{algorithm}

\begin{remark}
    \label{rem:goldwassersipser}
    The choice of the constant $2$ in the above protocol is arbitrary, and $2$ can be replaced by a slightly smaller constant, say $1.9$.
\end{remark}

We now state the main theorems of Koiran, giving interactive protocols for the problem of deciding whether a system of polynomial equations has a solution, and for the problem of estimating the dimension of an algebraic set.
We first state the main theorem of \cite{K96}, on the Hilbert nullstellensatz problem.

\begin{theorem}
    \label{thm: koiran nullstellensatz in am}
    Assume GRH.
    Given $f_{1}, \dots, f_{m} \in R$, there is an $\AM$ protocol that decides if $Z(f_{1}, \dots, f_{m}) \neq \emptyset$.
\end{theorem}

The above theorem was further generalised by Koiran in \cite[Theorem 4.1]{koiran1997randomized}, where now one wants to decide a lower bound on the dimension of the given algebraic set.

\begin{theorem}
    \label{thm: koiran dimension in AM}
    Assume GRH.
    Given $f_{1}, \dots, f_{m} \in R$, and an integer $r$, there is an $\AM$ protocol that decides if $\dim Z(f_{1}, \dots, f_{m}) \geq r$.
\end{theorem}

\begin{remark}
    We make two remarks on the proofs of \cref{thm: koiran nullstellensatz in am} and \cref{thm: koiran dimension in AM} that will be crucial to the way we invoke these results.
    The first of these is regarding the representation of input polynomials.
    The AM protocols for the above problems involve evaluating the polynomials $f_{1}, \dots, f_{m} \pmod p$ at points in $\bF_{p}$, where $p$ is a prime of bit complexity polynomial in the input.
    Therefore, the polynomials $f_{1}, \dots, f_{m}$ are allowed to be given either as white box circuits of polynomial size (here the size includes the bit complexity of the constants in the circuit), or more generally as black-boxes that allow mod $p$ queries.

    The second remark is regarding the parameters.
    Suppose $\deg{f_{i}} \leq d$ and $\lheight{f_{i}} \leq h$.
    The length of the messages and the computation done by Arthur in the protocols in \cref{thm: koiran nullstellensatz in am} and \cref{thm: koiran dimension in AM} is polynomial in $\log\br{{h \cdot 2^{\br{n \log \sigma}^{c}}}}$, for a universal constant $c$, where $\sigma := dm + 2$.
    Therefore, if we create a system of polynomials $g_{1}, \dots, g_{m}$ with $\deg{g_{i}} \leq d^{n}$ and $\lheight{h \cdot 2^{\br{n \log \sigma}^{c}}}$, then the protocols in \cref{thm: koiran nullstellensatz in am} and \cref{thm: koiran dimension in AM} applied to $g_{1}, \dots, g_{m}$ still run in time $\poly{n, m, d, h}$ as long as we ensure that $g_{1}, \dots, g_{m}$ have circuits of size $\poly{n, m, d, h}$.
    This will be crucial in our proofs.
\end{remark}

Lastly, we show that the ideal primality testing problem is $\coNP$-hard, even for the special case of equidimensional Cohen-Macaulay ideals.

\begin{proposition}[coNP-hardness of ideal primality testing]
    \label{proposition:conp hardness}
    Given a 3CNF $\Phi$ on $n$ variables, there exists a polynomial time algorithm whose output is a circuit $C_\Phi$ computing polynomials $f_{1}, \dots, f_{m} \in R$, along with their partial derivatives $\partial_{i} f_{j}$, such that the ideal $I := (f_1, \ldots, f_m)S \subset S$ has the following properties:
    \begin{itemize}
        \item $I$ is a radical, equidimensional Cohen-Macaulay ideal
        \item $I$ is prime if and only if $\Phi$ is unsatisfiable.
    \end{itemize}
\end{proposition}
\begin{proof}
    Let $x_{1}^{2} - x_{1}, x_{2}^{2} - x_{2}, \dots, x_{n}^{2} - x_{n}, g_{1}, \dots, g_{r}$ be the arithmetization of $\Phi$, where $r$ is the number of clauses in $\Phi$.
    The instance $\Phi$ is satisfiable if and only if the arithmetized system has a solution.

    The reduction algorithm proceeds as follows.
    First check if $\br{0, \dots, 0}$ is a solution to the above system.
    If it is, then the algorithm returns the system $f_{1} = x_{1} \br{x_{1} - 1}$, clearly there is a constant-sized circuit that computes $f_{1}, \partial_{i} f_{1}$.

    Suppose $\br{0, \dots, 0}$ is not a solution to the above system.
    Consider the set of polynomials $\bc{x_{i}^{2} - x_{i}}_{i \in [n]} \cup \bc{x_{j}g_{i}}_{i \in [r], j \in [n]}$.
    Since each $g_{i}$ has a small circuit, there is a small circuit that computes all the above polynomials and all their derivatives.
    Further, the set of zeroes of the above system is exactly those points on the boolean hypercube that correspond to a satisfying assignment of $\phi$, and also the point $\br{0, \dots, 0}$.

    In each of the above cases, by the Lemma in \cite[Page~310]{seidenberg1974constructions}, the ideal constructed is radical.
    In the first case the ideal is a complete intersection therefore CM.
    In the second case the ideal is zero dimensional, therefore CM.
    If $\Phi$ is satisfiable, then in either case the system created has at least two components and is therefore not prime.
    If $\Phi$ is not satisfiable, the system created has unique solution $\br{0, \dots, 0}$, thus the ideal is prime.
\end{proof}

\subsection{Algebraic circuits}

The inputs to our algorithms are given as algebraic circuits.
Algebraic circuits are natural models for computing polynomials.
They consist of a directed acyclic graph, with nodes marked with $+$ and $\times$.
The source nodes are marked either with variables $x_{i}$ or with the constant $1$.
Each internal node computes the natural polynomial: nodes marked with $+$ compute the sums of their inputs, and nodes marked with $\times$ compute the product.
The edges are marked with constants, if $(u, v)$ is an edge with constant $\alpha$, then the input to $v$ corresponding to $u$ is $\alpha f_{u}$, where $f_{u}$ is the polynomial computed at the node $u$.
We refer the reader to the excellent survey \cite{shpilka2010arithmetic} for more background on algebraic circuits.

Circuits are studied both as uniform and as non-uniform models of computation.
Further, they are studied both in the setting where all constants are considered to have size $1$, and also in the setting where the logarithmic height of the constants is part of the size of the circuit.
Since we are using circuits as inputs to a problem in the Turing machine model of computation, we are naturally in the setting where the size of circuit includes the logarithmic heights of the constants of the circuit.
We assume that the constants in the circuit are all integers.

\textbf{Size of a circuit:} The size of an algebraic circuit is the sum of the logarithmic heights of the constants on the edges of circuits.
We assume that every edge has a constant (potentially $1$), therefore the size of a circuit is a lower bound on the number of edges of the circuit.

\textbf{Evaluating a circuit modulo p:}
Given a circuit $C$, a prime $p$, and a point $\alpha \in\bF_{p}^{n}$, there is a polynomial time algorithm that computes the evaluations at $\alpha$ of all the mod $p$ reductions of all polynomials computed by $C$.
The algorithm recursively computes the evaluation $\pmod p$ of the inputs to a given node in the circuit, and then performs arithmetic in $\bF_{p}$ to compute the evaluation of the gate itself.

\textbf{Structural results:}
We need the following technical result \cite[Lemma~4.16]{burgisser2013completeness} that bounds the logarithmic heights and degrees of polynomials computed by circuits of size $s$.
The result follows by induction on the circuit, and using the fact that every coefficient in the circuit has logarithmic height at most $s$.
\begin{lemma}
    \label{lemma: circuit height bound}
    If $C$ is a circuit of size $s$ computing polynomials $f_{1}, \dots, f_{m} \in R$ then $\lheight{f_{i}} \leq 2^{2s}$ and $\deg f_{i} \leq 2^{s}$.
\end{lemma}

We also need the following result from \cite{baur1983complexity} that proves that given a circuit $C$, there is a circuit $C'$ of similar size that also computes the partial derivatives of the polynomials computed by $C$.
\begin{lemma}
    \label{lemma: derivatives of a circuit}
    If $C$ is a circuit of size $s$ computing polynomials $f_{1}, \dots, f_{m} \in R$ there is a circuit $C'$ of size $5sm$ that computes $f_{1}, \dots, f_{m}$ and all the partial derivatives $\partial_{i} f_{j}$.
\end{lemma}
\subsection{Results from linear algebra}

Cramer's rule is a well known explicit formula for the solution of a linear system of the form $Ax = b$, when $A$ is a $n \times n$ matrix with $\det(A) \neq 0$.
Under these conditions, the unique solution to the above system is given by $x_{i} = \frac{\det A_{i}}{\det A}$, where $A_{i}$ is the matrix obtained by replacing the $i^{th}$ column of $A$ with the vector $b$.
The following easy consequence shows that a similar formula exists for under and overdetermined systems, provided a solution is promised to exist.
Note that the lemma applies for matrices with coefficients in any domain.
\begin{lemma}
    \label{lem:generalcramer}
    Suppose $A$ is an $n \times m$ matrix, and suppose the linear system $Ax = b$ is guaranteed to have a solution.
    Then there exists a solution where each $x_{i}$ is either $0$ or of the form $\frac{\det M_{i}}{\det N}$, where $\det M_{i}, \det N$ are minors of the augmented matrix $A | b$.
\end{lemma}
\begin{proof}
    Let $r := \Rank A$, so $r \leq m, n$.
    The fact that $Ax = b$ has a solution is equivalent to the fact that the augmented matrix $\bs{A | b}$ also has rank $r$.
    After rearranging the columns, we can assume that the first $r$ columns of $A$ are linearly independent.
    Since $b$ lies in the column span of $A$, it also lies in the column span of the first $r$ columns of $A$.
    Equivalently, if we set the variables $x_{m-r+1}, \dots, x_{m}$ to $0$, the resulting system still has a solution.
    It suffices to show the result for the resulting system, therefore we can assume that $A$ has full column rank, so $r = m$.

    After further rearranging rows, we can assume that the first $r$ rows of $A$ are linearly independent, and the remaining rows are linear combinations of the first $r$ rows.
    The augmented matrix $\bs{A | b}$ also has rank $r$, therefore the last $n-r$ rows of $\bs{A | b}$ are linear combinations of the first $r$ rows.
    If we write $A'x = b'$ for the linear system consisting of the first $r$ rows of $Ax = b$, then any solution to $A'x = b'$ is also a solution to $Ax = b$.
    Applying Cramer's rule to $A'x = b'$ gives us the desired result.
\end{proof}

\subsection{Results from algebraic geometry}

In this subsection, we collect some useful basic results from algebraic geometry which we will need in the later sections.
We begin by stating a characterisation of the dimension of an algebraic set.

\begin{lemma}
    \label{lem:dimensionelimination}
    Suppose $I \subset S$ is an ideal, and $V$ is the zeroset of $I$.
    Then $V$ has dimension $r$ if and only if the following two conditions hold.
    \begin{itemize}
        \item There is some subset $U \subset \bs{n}$ of size $r$ such that the elimination ideal $I_{U}$ is zero.
        \item For every subset $U' \subset \bs{n}$ of size $r+1$, the elimination ideal $I_{U'}$ is nonzero.
    \end{itemize}
\end{lemma}

The first of the above conditions is equivalent to $\dim V \geq r$, and the second condition is equivalent to $\dim V \leq r$.
We now state some standard facts about tangent spaces, which can be found in \cite[Chapter~2]{shafarevich1994basic}.

Let $I$ be a radical ideal generated by $f_{1}, \dots, f_{m}$.
Let $\cJ$ be the Jacobian of $I$, which is defined to be the matrix with $\cJ_{ij} = \partial f_{i} / \partial x_{j}$.
At any point $x \in V$, the tangent space $T_{x}(V)$ is isomorphic to the kernel of $\cJ(x)$, where $\cJ(x)$ is the Jacobian entrywise evaluated at the point $x$ \cite[Chapter~2, Section~1, Theorem~2.1]{shafarevich1994basic}.
At every nonsingular point $x$, the dimension of the tangent space is exactly equal to dimension of the component of $V$ passing through $x$ \cite[Chapter~2, Section~1, Theorem~2.3]{shafarevich1994basic}.
At every singular point $x$, the dimension of the tangent space is greater than the dimension of the components of $V$ passing through $x$.
The singular locus of $V$ is a proper algebraic subset of $V$, and does not contain any irreducible component of $V$.

We deduce that there is a point $x_{0} \in V$ such that $\dim T_{x_{0}}(V) < r$ if and only if $V$ has an irreducible component of dimension less than $r$.
The following lemma is an equivalent formulation of the above statement.

\begin{lemma}\label{lem:tangent dimension}
    If $V$ is a algebraic set of dimension $r$, and if $\cJ$ is the Jacobian of a radical ideal $I$ such that $V = Z(I)$, then $\Rank \cJ(x_{0}) > n - r$ for some $x_{0} \in V$ if and only if $V$ has an irreducible component of dimension less than $r$.
    
\end{lemma}

The following statement is \cite[Theorem 1]{heintz1983definability}, and is referred to as Bézout’s inequality.

\begin{theorem}[\bezout's inequality]\label{theorem: bezout inequality}
    Let $\bK$ be an algebraically closed field and $X, Y \subseteq \bA_\bK^n$ be locally closed subsets.
    Then, we have 
    $$\deg(X \cap Y) \leq \deg X \cdot \deg Y.$$
\end{theorem}

A careful inductive application of the Bézout’s inequality implies the following inequality, which we also refer to as Bézout’s inequality.
The statement given below is a special case of \cite[Proposition~2.3]{heintz1980testing}.

\begin{theorem}[Degree bounds for closed sets]\label{theorem: degree bounds}
    Let $\bK$ be an algebraically closed field and $f_1, \dots, f_m \in \bK[x_1, \dots, x_n]$ be polynomials of degree at most $d$. 
    Then, we have 
    $$ \deg V(f_1, \ldots, f_m) \leq d^{\min\br{m, n}}.$$
\end{theorem}

The following lemma gives effective bounds on hyperplane sections that reduce the dimension of projective varieties, a proof can be found in \cite[Lemma~19]{guo2019algebraic}
\begin{lemma}\label{lem:intersectdimensiondrop}
    Suppose $V \subset \bP^{n}$ is a projective zeroset of dimension $r$ and degree $D$.
    If $\ell$ is a linear form where each coefficient is chosen uniformly and independently from a set $B \subset \bN$, then $\dim V \cap Z(\ell) = r-1$ with probability at least $1 - D/\abs{B}$.
    If $\ell_{1}, \dots, \ell_{r+1}$ are linear forms chosen the same way then $V \cap Z(\ell_{1}, \dots, \ell_{r+1}) = \emptyset$ with probability at least $1 - (r+1) D / \abs{B}$.

    Suppose $W \subset \bA^{n}$ is an affine zeroset of dimension $r$ and degree $D$.
    If $\ell$ is a linear polynomial where each coefficient is chosen uniformly and independently from a set $B \subset \bN$, then $\dim W \cap Z(\ell) = r-1$ with probability at least $1 - 2D/\abs{B}$.
    If $\ell_{1}, \dots, \ell_{r+1}$ are linear polynomials chosen the same way then $W \cap Z(\ell_{1}, \dots, \ell_{r+1}) = \emptyset$ with probability at least $1 - 2(r+1) D / \abs{B}$.
\end{lemma}

The following is a sufficient condition for a linear map to be a Noether normalising map.
This statement, and its proof can be found in \cite[Chapter~1, Section~5, Theorem~1.15]{shafarevich1994basic}.

\begin{lemma} \label{lem:noethernormalising}
    Suppose $V \subset \bP^{n}$ is a projective variety disjoint from a $k$-dimensional linear subspace $E \subset \bP^{n}$.
    Then the projection $\pi: X \to \bP^{n-k-1}$ with center $E$ defines a finite map $X \to \pi(X)$.
\end{lemma}

Combining the two lemmas above gives us the following effective Noether Normalisation theorem.

\begin{lemma}[Effective Noether Normalisation] \label{lem:effectiveNN}
    Suppose $V \subset \bP^{n}$ is a projective zeroset of dimension $r$ and degree $D$.
    If $\ell_{1}, \dots, \ell_{r+1}$ are linear forms where each coefficient is chosen uniformly and independently from a set $B \subset \bN$, then the map $\pi: V \to \bP^{r}$ with coordinate functions $\ell_{i}$ is a well defined finite map with probability at least $1 - (r+1) D / \abs{B}$.

    Suppose $W \subset \bA^{n}$ is an affine zeroset of dimension $r$ and degree $D$.
    If $\ell_{1}, \dots, \ell_{r}$ are linear polynomials where each coefficient is chosen uniformly and independently from a set $B \subset \bN$, then the map $\pi: W \to \bA^{r}$ with coordinate functions $\ell_{i}$ is a well defined finite map with probability at least $1 - 2 (r+1) D / \abs{B}$.
\end{lemma}

\begin{proof}
    For the projective case, by \cref{lem:intersectdimensiondrop}, the zeroset $Z(\ell_{1}, \dots, \ell_{r+1})$ is a linear space disjoint from $V$ with probability at least $1 - (r+1) D / \abs{B}$.
    If this holds, then $\pi$ is a finite and well defined map by \cref{lem:noethernormalising}.
    The affine case follows by applying the projective case to the closure of $W_{p}$, and picking $\ell_{r+1} = x_{0}$.
\end{proof}

Finally we state a necessary and sufficient condition for equidimensional Cohen-Macaulay rings to  be reduced.
The following statement is a special case of \cite[Item~a, Theorem~18.15]{eisenbud2013commutative}.
For the readers convenience, we isolate the ingredients in the proof here.
\begin{theorem}
    \label{theorem: serre jacobian}
    Let $I = \ideal{f_{1}, \dots, f_{m}} \subset \bC\bs{x_{1}, \dots, x_{n}}$ be a codimension $c$ ideal such that $\bC\bs{x_{1}, \dots, x_{n}} / I$ is equidimensional Cohen-Macaulay.
    Let $\cJ$ be the Jacobian of $I$, that is, the $m \times n$ matrix with entries $\cJ_{ij} = \partial f_{i} / \partial x_{j}$.
    Let $J$ be the ideal generated by the $c \times c$ minors of $\cJ$ in the ring $\bC\bs{x_{1}, \dots, x_{n}} / I$.
    Then $I$ is radical if and only if $J$ has codimension at least $1$.
\end{theorem}
\begin{proof}
    For any Noetherian ring, being reduced is equivalent to the properties $R_{0}$ and $S_{1}$ \cite[\href{https://stacks.math.columbia.edu/tag/031R}{Tag 031R}]{stacks-project}.
    Cohen-Macaulay rings satisfy $S_{k}$ for every $k$ \cite[\href{https://stacks.math.columbia.edu/tag/0342}{Tag 0342}]{stacks-project}, therefore it suffices to show that $R_{0}$ is equivalent to $J$ having codimension at least $1$.
    By the Jacobian criterion \cite[Theorem~16.19]{eisenbud2013commutative}, for any prime $\kp$ of $\bC\bs{x_{1}, \dots, x_{n}} / I$, regularity of the localisation at $\kp$ is equivalent to $J \not \subset \kp$.
    Finally, $J$ has codimension at least $1$ if and only if it is not contained in any minimal prime, which by the above happens if and only the localisation at every minimal prime of $\bC\bs{x_{1}, \dots, x_{n}} / I$ is regular, which is precisely the condition $R_{0}$.
\end{proof}

\subsection{Results from number theory} \label{subsec:number theory prelims}

We state some algebraic number theory facts that we will need.
We do not provide any proofs since these facts are standard, \cite[Chapter~2, 3]{milneANT} is an excellent exposition of all of these results.

Let $q \in \bZ\bs{z}$ be a monic irreducible polynomial with $\deg q = e$, and let $\alpha$ be a root of $q$.
Let $\bF := \bQ\br{\alpha}$ be the algebraic extension of $\bQ$ generated by $\alpha$, we have $\bF \cong \bQ\bs{z} / \ideal{q}$.
Let $O_{\bF}$ be the ring of integers of $\bF$, that is, the set of elements $\beta \in \bF$ such that $\beta$ satisfies a monic equation with coefficients in $\bZ$.
We have $\bZ \subset O_{\bF}$, and $O_{\bF}$ is integrally closed in $\bF$.
The discriminant of $q$, denoted $\disc{z}{q}$ is defined to be $\resultant{q}{\partial q}{z}$.

Usually the ring of integers $O_{\bF}$ is different from $\bZ\bs{\alpha}$.
We have $\bZ\bs{\alpha} \subset O_{\bF}$, but this inequality might be strict.
However, it holds that $O_{\bF} \subset (1 / \disc{z}{q}) \bZ\bs{\alpha}$.
Therefore we can represent any algebraic integer as an element of $\bQ\bs{z}$, where the coefficients have common denominator $\disc{z}{q}$.

The ring $O_{\bF}$ is a Dedekind domain, and every nonzero prime ideal of $O_{\bF}$ is maximal.
For any prime $\kp \subset O_{\bF}$, there is a unique prime $p \in \bN$ such that $\kp \cap \bZ = \ideal{p}$.
If this happens, then we say $\kp$ lies above $p$.
While $O_{\bF}$ is not factorial in general, it admits unique factorisation of ideals.
Every ideal $I \subset O_{\bF}$ can be the written uniquely as a product of primes, $I = \prod \kp_{i}^{e_{i}}$.
In particular, the ideal $\ideal{p}$ for a prime $p \in \bN$ itself can be factored as $\ideal{p} = \prod \kp_{i}^{e_{i}}$.
The set of primes ideals of $O_{\bF}$ occurring in this factorisation of $\ideal{p}$ are exactly the prime ideals of $O_{\bF}$ that lie above $p$.
When $p \not | \disc{z}{q}$, $e_{i} = 1$ for all $i$, if this holds we say that $p$ is unramified.

For any unramified $p$, and $\kp$ lying above $p$, the quotient map $O_{\bF} \to O_{\bF} / \kp$ is well structured.
Since $\kp$ is maximal, and since $p \in \kp$, the quotient $O_{\bF} / \kp$ is a field of characteristic $p$, and is in fact a finite extension of $\bF_{p}$.
The prime $\kp$ corresponds to an irreducible factor $q_{1}$ of $q \pmod p$, such that $\kp = pO_{\bF} + q_{1}(\alpha) O_{\bF}$.
It follows that $O_{\bF} / \kp \cong \bF_{p}\bs{z} / q_{1}(z)$.
Composing with the inclusion map $\bZ\bs{\alpha} \to O_{\bF}$ gives us a map $\bZ\bs{\alpha} \to \bF_{p}\bs{z} / q_{1}(z)$.
We usually apply this to the case when $q_{1}$ is a linear polynomial, and therefore $\bF_{p}\bs{z} / q_{1}(z) = \bF_{p}$.
This map therefore extends the usual map $\bZ \to \bF_{p}$.

We now state a corollary of Gauss's lemma.
The following formulation is from \cite[Lemma~7.1]{weinberger1976factoring}, we continue to refer to it as Gauss's lemma.

\begin{lemma}[Gauss's lemma]
    \label{lem:gauss's lemma}
    Suppose $\delta \in O_{\bF}$ and $f \in (1/\delta)O_{\bF}$ is a monic polynomial with factorisation $f = gh$ in $\bF$.
    If $g, h$ are monic then $g, h \in (1/\delta)O_{\bF}$.
\end{lemma}

We will need three more technical results from number theory.
The first of these is an effective Lang-Weil bound.
The Lang-Weil bounds are upper and lower bounds on the number of $\bF_{p}$ points on irreducible varieties that are $\bF_{p}$ definable.
Classically these bounds have error terms, and effective versions of such bounds give bounds on the coefficients of the error terms.
The following is a version of such a bound from \cite[Theorem~7.1]{cafure2006improved}.
While we state the bound in generality, we only require it for curves.
\RO{maybe put a note here that we could get the full version for reducible algebraic sets simply by putting the number of $\bF_p$-definable components... for after the deadline...}

\begin{theorem}[Effective Lang-Weil bound]
    \label{thm:effective lang weil}
    Let $V \subset \bF_{p}^{n}$ be an absolutely irreducible affine variety defined over $F_{p}$ of dimension $r$ and degree $D$.
    Suppose also that $p > 2(r+1)D^{2}$.
    Then
    $$ \abs{V(\bF_{p}) - p^{r}} < (D-1)(D-2) p^{r-1/2} + 5D^{13/3}p^{r-1}.$$
\end{theorem}

The second number theoretic statement we need is a lower bound on the number of primes $p$ for which $q \pmod p$ has a root in $\bF_{p}$.
Observe that the Chebotarev density theorem states that this fraction of primes is the same as the fraction of the Galois group of $q$ that has at least one fixed point, and this fraction is at least $1 / e$.
We require an effective version of this statement.
The following statement is from \cite[Corollary~1]{K96}.
The proof itself is based on a bound by \cite{adleman1983irreducibility} which in turn is based on an effective version of the Chebotarev density theorem \cite{lagarias1977effective}.
We note that this last result is conditional, and assumes the GRH, and this is also what makes our result conditional.

\begin{theorem}
    \label{lem:effective chebotarev}
    Assuming the GRH holds, there exists an absolute constant $c$ such that
    $$ \pi_{q}(x) \geq \frac{1}{e} \br{\pi(x) - \log \disc{z}{q} - c x^{1/2} \log\br{\disc{z}{q} x^{e}}},$$
    where $\pi(x)$ is the prime counting function and $\pi_{q}(x)$ is the number of primes $p \leq x$ such that $q \pmod p$ has a root in $\bF_{p}$.
\end{theorem}

Finally, we need an unconditional lower bound on the prime counting function $\pi(x)$.
One such estimate is the following classical bound.
\begin{theorem}[{\cite[Corollary~1]{rosser1962approximate}}]
    \label{thm:unconditional prime counting}
    For $x \geq 17$ we have $\pi(x) > x / \ln{x}$.
\end{theorem}

\subsection{Degree bounds for polynomial ideals}

We now recollect some important complexity results for polynomial ideals.
We begin with Jelonek's effective Nullstellensatz, from \cite[Theorem 1.1]{jelonek2005effective}.

\begin{theorem}[Effective Nullstellensatz]\label{theorem: jelonek nullstellensatz}
    Let $\bK$ be an algebraically closed field, and $X \subset \bK^m$ be an affine $n$-dimensional variety of degree $D$. 
    Let $f_1, \ldots, f_m \in \bK[X]$ be non-constant polynomials without common zeros, where $d_i := \deg f_i$, and $d_1 \geq \cdots \geq d_m$.
    Lastly, let 
    $$
    N(d_1, \dots, d_m; n) := 
    \begin{cases}
        \prod_{i=1}^m d_i \quad \text{ if } n > 1 \text{ and } n \geq m \\
        d_m \cdot \prod_{i=1}^{n-1} d_i \quad \text{ if } m > n > 1 \\
        d_1 \quad \text{ if } n=1
    \end{cases}
    $$
    There exist polynomials $g_i \in \bK[X]$ with 
    $$
    \deg(f_i g_i) \leq 
    \begin{cases}
        D \cdot N(d_1, \ldots, d_m ; n), \quad \text{ if } m \leq n \\
        2D \cdot N(d_1, \dots, d_m, n) - 1, \quad \text{ if } m > n
    \end{cases}
    $$
    such that $1 = \sum_{i=1}^m f_i g_i$.
\end{theorem}

We now state some useful bounds on the complexity of representation of the reduced \grobner{} basis of any ideal.
The bounds that we state below are dependent on the Krull dimension of the given ideal\footnote{as is standard in commutative algebra, the Krull dimension of an ideal $I$ is the Krull dimension of the quotient ring $\bK\bs{X} / I$.}, and this will be crucial in our applications in the later sections.
These bounds are from \cite[Theorem 4]{mayr2013dimension}.

\begin{theorem}[\grobner{} basis complexity]\label{theorem: groebner basis complexity}
    Let $\bK$ be an infinite field and $I = (f_1, \dots, f_m) \subseteq \bK[x_1, \dots, x_n]$ be an ideal of dimension $r$, where $d_j := \deg(f_j)$ and $d_1 \geq \cdots \geq d_m$.
    Let 
    $$ B := 2 \cdot \br{\dfrac{1}{2} \br{\br{d_1 \cdots d_{n-r}}^{2(n-r)} + d_1} }^{2^r}. $$
    Given any admissible monomial ordering, if $G = \{g_1, \ldots, g_t\}$ is the reduced \grobner{} basis of $I$, then we have: 
    $$ \deg(g_i) \leq B \quad \forall i \in [t].  $$
    Moreover, for each $i \in [t]$, there are polynomials $h_{i1}, \dots, h_{im} \in \bK[x_1, \ldots, x_n]$ satisfying $\deg(f_j \cdot h_{ij}) \leq B$ such that
    $$ g_i = \sum_{j=1}^m f_j h_{ij}. $$
\end{theorem}

\noindent The above bounds, when combined with the division algorithm, allow us to derive the following upper bound on the dimension of the linear system (over the base field) needed to check membership in polynomial ideals.

\begin{corollary}[Representation degree]\label{corollary: representation degree ideal membership} 
    Let $\bK$ be an infinite field and $I = (f_1, \ldots, f_m) \subseteq \bK[x_1, \dots, x_n]$ be an ideal of dimension $r$, where $d_j := \deg(f_j)$ and $d_1 \geq \cdots \geq d_m$.
    Let 
    $$ B := 2 \cdot \br{\dfrac{1}{2} \br{\br{d_1 \cdots d_{n-r}}^{2(n-r)} + d_1} }^{2^r}.$$
    Given any polynomial $f \in \bK[x_1, \dots, x_n]$ with $\deg(f) = D$, we have that $f \in I$ if, and only if, there exist $g_1, \dots, g_m \in \bK\bs{x_1, \dots, x_n}$ such that $\deg(f_j g_j) \leq \max\br{B, D}$ and 
    $\displaystyle f = \sum_{j=1}^m f_j g_j$.
\end{corollary}

\subsection{An Effective Bertini Theorem}

Bertini's second theorem states that the intersection of an irreducible variety of dimension $r \geq 2$ with a random hyperplane of dimension $n-r+1$ is an irreducible curve.
This statement allows us to reduce the problem of irreducibility testing of arbitrary high dimensional varieties to that of curves and points.
The theorem plays an important role in many routines in computer algebra, a comprehensive survey of its applications can be found in \cite[Section~9.1.3]{dickenstein2005solving}.
The effective versions of the theorem are usually stated for hypersurfaces, reduction from the general case to hypersurfaces follows from a projection argument.

In order to apply the theorem to arbitrary varieties, a preprocessing projection step is applied.
The following two results are from \cite{cafure2006improved}, where the effective Bertini theorem is used to obtain improved Lang-Weil bounds.
While the results in \cite{cafure2006improved} are stated for varieties over $\bF_{q}$, the following two results also hold for varieties over any perfect field, with the same proofs.

\begin{lemma}[{\cite[Proposition~6.1, Proposition~6.3]{cafure2006improved}}]
    \label{lem:effectivebirational}
    Let $V$ be an equidimensional affine variety of dimension $r$ and degree $D$.
    Let $\Lambda$ be a $(r+1) \times n$ matrix of variables, and let $\Gamma$ be a $r+1$ dimensional vector of variables.
    There exists a nonzero polynomial $G \in \overline{\bQ}\bs{\Lambda, \Gamma}$ of degree at mot $2(r+1)D^{2}$ such that for any values $\lambda, \gamma$ of the variables $\Lambda, \Gamma$ satisfying $G(\lambda, \gamma) \neq 0$, the projection $\pi$ defined by $\ell_{i} = \sum \lambda_{ij} x_{j} + \gamma_{j}$ is a birational map between $V$ and its image $\pi(V)$.
    Further, there is a polynomial $g$ of degree $D$ such that $\pi(V) \setminus Z(g)$ and $V \setminus \pi^{-1}(Z(g))$ are isomorphic.
\end{lemma}

\begin{lemma}[{\cite[Corollary~3.4]{cafure2006improved}}]
    \label{lem:effectivebertinihypersurface}
    Suppose $f \in \bZ\bs{x_{1}, \dots, x_{n}}$ is an absolutely irreducible polynomial of degree $D$.
    For a point $(v, w, z) \in \overline{\bQ}^{n} \times \overline{\bQ}^{n-1}\times \overline{\bQ}^{n-1}$, let $f_{v, w, z}(x, y) := f(x + v_{1}, w_{2} x + z_{2} y + v_{2}, \dots, w_{n} x + z_{n}y + v_{n})$.
    There exists a polynomial $H$ of degree at most $10D^{5}$ such that for any $v, w, z$ satisfying $H(v, w, z) \neq 0$, the polynomial $f_{v, w, z}$ remains absolutely irreducible.
\end{lemma}

\begin{corollary}[Effective Bertini Theorem]
    \label{cor:effectivebertini}
    Suppose $V$ is an equidimensional affine variety of dimension $r$ and degree $D$.
    There is a randomised algorithm that returns linear equations $\ell_{1}, \dots, \ell_{r-1}$ with $\lheight{\ell_{i}} \leq (nD)^{c}$ such that the following holds with high probability: $V \cap Z(\ell_{1}, \dots, \ell_{r-1})$ is an equidimensional affine curve, with the same number of irreducible components as $V$.
\end{corollary}

\begin{proof}
    We focus on each irreducible component of $V$ and apply a union bound at the end, therefore assume without loss of generality that $V$ is irreducible.
    Fix set $B := \bs{10D^{7}}$, a subset of $\bN$.
    Pick linear polynomials $h_{1}, \dots, h_{r+1}$ with coefficients from $B$, and let $\pi$ be the projection map with coordinate functions $h_{i}$.
    By \cref{lem:effectivebirational}, $\pi(V)$ is birational with its image with probability at least $1 - 2(r+1)D^{2} / \abs{B}$.
    The image is an irreducible hypersurface of degree $D$, say $Z(f)$.
    Now pick a point $(v, w, z)$ from $B^{n} \times B^{n-1} \times B^{n-1}$.
    By \cref{lem:effectivebertinihypersurface}, $f_{v, w, z}$ is absolutely irreducible with probability at least $1 - 10D^{5}/\abs{B}$.
    Pick a linear subspace $L \subset \bA^{r+1}$ such that $f_{v, w, z}$ is isomorphic to $f|_{L}$, this can be performed by elementary linear algebra.
    By construction, $L \cap \pi(V)$ is an irreducible curve.
    Since $L$ is a random linear subspace, $(\pi(V) \cap Z(g)) \cap Z(L)$ is a finite set, where $g$ is the polynomial guaranteed by \cref{lem:effectivebertinihypersurface} \footnote{Note that \cref{lem:intersectdimensiondrop} does not apply since the distribution of $L$ is different, but the same proof works}.
    Now we show that $\pi^{-1}(L)$ is the required subspace.

    First observe that $\pi^{-1}(L)$ has defining equations of logarithmic height $\poly{D, n}$, this is because the equations are obtained by inverting matrices with entries of logarithmic height $\poly{D, n}$.
    Now we know that $\pi(V) \cap Z(L)$ is an irreducible curve.
    If $\psi$ is the inverse of the birational map $\pi$, then $\pi^{-1}(L) \cap V$ is exactly $\overline{\psi(\pi(V) \cap Z(L))}$, which is irreducible and dimension $1$.
\end{proof}

%% file: height-alg-int.tex
In this section we recall some height bounds for certain operations over polynomial rings over the integers, as well as of certain operations over polynomial rings over certain algebraic numbers.
Then, we proceed to prove some height bounds for membership problems in given ideals, which is done in \cref{subsection: height bounds ideals}.
We then conclude the section with some lemmas on absolutely irreducible factors of bivariate polynomials, where we discuss the work of Kaltofen and its implications to our setting.

\RO{we should probably have a discussion about heights here again... polish this after the deadline}

\subsection{Height bounds for elementary operations}

The following are some elementary height bounds on the results of basic arithmetic operations performed on polynomials.
The next two statements are from \cite[Lemma~1.2]{kps01}.

\begin{lemma}
    \label{lem:kpsheightbound}
    Let $f_{1}, \dots, f_{m} \in R$ be polynomials with $\deg{f_{i}} \leq d$ and $\lheight{f_{i}} \leq h$.
    Then 
    \begin{enumerate}
        \item $\lheight{\sum f_{i}} \leq h + \log\br{m}$.
        \item $\lheight{\prod f_{i}} \leq hm + md \log\br{n+1}$.
        \item If $g \in \bZ\bs{y_{1}, \dots, y_{m}}$ then $\lheight{g(f_{1}, \dots, f_{m})} \leq \lheight{g} + \deg g\br{h + \log\br{m+1} + d \log\br{n+1}}$.
    \end{enumerate}
\end{lemma}

The above lemma allows us to get height bounds for the determinant of a matrix with ring elements.

\begin{corollary}
    \label{cor:detheightbound}
    Let $M$ be a $m \times m$ matrix with $M_{ij} \in R$ such that $\deg(M_{ij}) \leq d$ and $\lheight{M_{ij}} \leq h$ for all $i, j$.
    Then $\lheight{\det M} \leq m\br{h + \log\br{m} + d\log\br{n+1}}$.
\end{corollary}

Once we have the above height bounds, the next corollary yields height bounds on the resultant and on the cofactors of the resultant identity.

\begin{corollary}
    \label{cor:resdischeightbounds}
    Let $f, g \in \bZ\bs{y}$ be coprime polynomials with $\deg{f}, \deg{g} \leq d$ and $\lheight{f}, \lheight{g} \leq h$.
    Then there exists $a, b \in \bZ\bs{y}$ with $\deg{a} < \deg{g}$ and $\deg{b} < \deg{f}$ such that $\resultant{f}{g}{y} = af + bg$, where 
    $$\lheight{a}, \lheight{b}, \lheight{\resultant{f}{g}{y}} \leq 2d\br{h + \log\br{2d + 1}}.$$
    In particular, if $\Delta$ is the discriminant of $f$ then $\lheight{\Delta} \leq 2d\br{h + \log\br{d} + \log\br{2d + 1}}.$
\end{corollary}
\begin{proof}
    In $\bQ\bs{y}$ we have the Bézout identity $sf + tg = 1$, with $\deg{s} < \deg{g}$ and $\deg{t} < \deg{f}$.
    We can write this as a linear system of at most $2d$ equations in at most $2d$ variables.
    The solution to this system is given by Cramer's rule (\cref{lem:generalcramer}): each coefficient of $s, t$ has numerator given by the determinant of a matrix of size $2d \times 2d$ with entries the coefficients of $f, g$.
    The denominator is common among all the coefficients, and is similarly given by such a determinant.
    In fact this denominator is exactly $\resultant{f}{g}{y}$.
    Clearing common denominators from the equation gives us the claimed identity with the claimed bounds.
    The last statement follows from the fact that the discriminant is exactly $\resultant{f}{f'}{y}$, and $\lheight{f'} \leq h + \log{d}$.
\end{proof}

We also need to bound the logarithmic height of any rational root of an integral polynomial.

\begin{lemma}
    \label{lem:intunivarrationalroot}
    Let $f \in \bZ\bs{y}$ be a degree $d$ polynomial and suppose $a/b$ is a rational root of $f$ in minimal form.
    Then $\lheight{a}, \lheight{b} \leq \lheight{f}d + 1$.
\end{lemma}
\begin{proof}
    Suppose $f = \sum f_{i} y^{i}$, where $f_i \in \bZ$.
    Multiplying throughout by $f_{d}^{d-1}$, and replacing $f_{d}y$ with a variable $x$ gives us a monic integer polynomial $g = x^{d} + \sum_{i=1}^{d-1} g_{i} x^{i}$ with $\deg{g} = d$ and $\lheight{g} \leq \lheight{f}d$.
    Since $a/b$ is a root of $f$, we have tht $c := f_{d} a / b$ is a root of $g$.
    Further, $c \in \bZ$ since $g$ is a monic integer polynomial.
    By \cite[Theorem~2]{mignotte83} we have $\lheight{c} \leq 1 + \lheight{g} \leq \lheight{f}d + 1$.
    We can obtain $a / b$ by putting $c / f_{d}$ in minimal form, therefore $\lheight{a}, \lheight{b} \leq \lheight{c}, \lheight{f_{d}} \leq \lheight{f}d + 1$.
\end{proof}

We can also bound the heights of the quotient and the remainder upon division by a monic polynomial, as the following proposition shows.

\begin{proposition}
    \label{proposition: height univariate division with remainder}
    Let $f, g \in \bZ[y]$ be polynomials where $d := \deg f \geq \deg g =: e$, $\lheight{f}, \lheight{g} \leq h$, and $g$ is a monic polynomial.
    If $f = g \cdot q + r$, where $g, r \in \bZ[y]$ with $\deg r < e$, then we have $\lheight{g}, \lheight{r} \leq (d+1) \cdot (h + 2 \log{(d+1)})$.
\end{proposition}

\begin{proof}
    Note that $g, r$ are the unique solution to the following system of linear equations: 
    $N \cdot \begin{pmatrix} g \\ r \end{pmatrix} = \begin{pmatrix} f \end{pmatrix}$,
    where $N$ is a $(d+1)\times (d+1)$ matrix of the form $N = \begin{pmatrix} M_q & P\end{pmatrix}$ and $P = \begin{pmatrix}
        0 \\ I_{e}
    \end{pmatrix}.$
    Since $q$ is monic, we know that $N$ is a unipotent lower triangular matrix.
    Thus, $\det(N) = 1$, and by Cramer's rule (\cref{lem:generalcramer}) we have that each coefficient of $g$ and $r$ is given by the determinant of a matrix in $\bZ^{(d+1) \times (d+1)}$ with height $\leq h$.
    By \cref{cor:detheightbound} we have the desired bound.
\end{proof}

As discussed earlier, we will often have to work with polynomials with coefficients lying in some finite extension generated by an algebraic integer $\alpha$.
If $\alpha$ has monic minimal polynomial $q(z) \in \bZ\bs{z}$ with $\deg q(z) \leq e$ then we will represent elements of $A = \bZ\bs{\alpha}\bs{x_{1},\dots, x_{n}} \simeq R[z]/(q)$ as polynomials in $R[z]$ of degree less than $e$ in $z$.
Thus, elements of $A$ inherit the logarithmic height function from $R[z]$.

After some algebraic operations, for example taking products of such polynomials, the resulting polynomial will no longer have degree less than $e$ in the variable $z$, and we must perform a step of going modulo $q(z)$.
The following lemma bounds the change in logarithmic height by this operation.
A sharper version of this lemma can be found in \cite[Lemma~1]{kaltofen1995effective}.

\begin{lemma}
    \label{lem:modgheightbounds}
    Let $q(z) \in \bZ\bs{z}$ be a monic polynomial with $\deg{q} = e$.
    Suppose $f \in R\bs{z}$ is a polynomial with $\deg_{z}{f} = d \geq e$.
    Lastly, suppose $\lheight{f}, \lheight{q} \leq h$.
    Then $\lheight{f \bmod q} \leq (d + 1) \br{h + 2 \log\br{d+1}}$.
\end{lemma}

\begin{proof}
    Let $\cM_f$ be the set of exponent vectors of $f$, when viewed as a polynomial in $\bZ[z][x_1, \ldots, x_n]$.
    Hence, we have $f = \displaystyle \sum_{\vec{a} \in \cM_f} f_\vec{a}(z) \cdot \vec{x}^\vec{a}$.
    Therefore, $f \bmod q = \displaystyle \sum_{\vec{a} \in \cM_f} (f_\vec{a}(z) \bmod q) \cdot \vec{x}^\vec{a}$ and thus 
    $$\lheight{f \bmod q} = \max_{\vec{a} \in \cM_f} \lheight{f_\vec{a}(z) \bmod q} \leq (d+1) \cdot (h + 2\log{(d+1)}).$$
    Where the last inequality follows from \cref{proposition: height univariate division with remainder}.
\end{proof}

Combining \cref{cor:detheightbound} with \cref{lem:modgheightbounds} we are able to obtain the following height bounds for determinants with entries in rings of the form $\bZ[\alpha]$, where $\alpha$ is an algebraic integer.

\begin{corollary}
    \label{corollary: determinant height bounds algebraic integer}
    Let $q(z) \in \bZ[z]$ be monic with $\deg q = e$, $\lheight{q} \leq h$ and let $\alpha \in \overline{Q}$ be a root of $q$.
    If $M \in \bZ[\alpha]^{m \times m}$ is such that $\lheight{M_{ij}} \leq h$ for all $i,j \in [m]$, we have that $\lheight{\det M} \leq e \cdot m  \cdot \br{ m \cdot (e + h + \log m) + 2 \log(em)  }$.
\end{corollary}

\begin{proof}
    By regarding $M \in \bZ[z]^{m \times m}$, with $\deg{M_{ij}} < e$ and $\lheight{M_{ij}} \leq h$ for all $i,j \in [m]$ \cref{cor:detheightbound} implies that $\deg \det M \leq (e-1) \cdot m$ and $\lheight{\det M} \leq m \cdot (h + \log(m) + (e-1))$.
    Now, applying \cref{lem:modgheightbounds} to $\det{M}$ with quotient $q(z)$, we have that 
    \begin{align*}
        \lheight{(\det{M}) \bmod q} 
        &\leq ((e-1) \cdot m + 1) \cdot \br{m \cdot (h + \log(m) + (e-1)) + 2 \log((e-1)m + 1)} \\
        &\leq e \cdot m \cdot \br{m \cdot (e + h + \log(m)) + 2 \log(em)}. \qedhere
    \end{align*}
\end{proof}

The following corollary obtains height bounds for basic operations over the ring $A$.

\begin{corollary}
    \label{corollary: height bounds polynomials algebraic number coefficients}
    Let $q(z) \in \bZ[z]$ be monic with $\deg q = e$, $\lheight{q} \leq h$ and let $\alpha \in \overline{Q}$ be a root of $q$.
    Let $A := \bZ[\alpha][x_1, \dots, x_n]$.
    Let $f_{1}, \dots, f_{m} \in A$ be polynomials with $\deg{f_{i}} \leq d$ and $\lheight{f_{i}} \leq h$.
    Then 
    \begin{enumerate}
        \item $\lheight{\sum_{i=1}^m f_{i}} \leq h + \log\br{m}$.
        \item $\lheight{\prod_{i=1}^m f_{i}} \leq em \cdot (hm + dm \log(n+2) + 2 \log(em))$.
        \item If $g \in \bZ[\alpha]\bs{y_{1}, \dots, y_{m}}$ with $d_g := \deg(g)$ then
    $$\lheight{g(f_1, \dots, f_m)} \leq 2ed_g \cdot \bs{\lheight{g} + d_g \cdot (h + \log(m+1) + \max\bc{e, d} \cdot \log(n+2)) + 2\log(2ed_g)}$$
    \end{enumerate} 
\end{corollary}

\begin{proof}
    For item 1, note that the coefficients of $\sum_{i=1}^m f_i$ are still given by polynomials in $\bZ[z]$ with degree less than $e$, since addition does not increase the degree. 
    Thus, the bounds follow from \cref{lem:kpsheightbound} item 1.

    For item 2, let $p = \prod_{i=1}^m f_i$.
    In this case, we have that $p = \hat{p} \bmod q$, where $\hat{p} \in R[z]$ is the product of the polynomials $f_i$ when seen as elements of $R[z]$.
    By \cref{lem:kpsheightbound} item 2, $\lheight{\hat{p}} \leq h \cdot m + md\log(n+2)$.
    If $\deg_z(\hat{p}) < e$, we have $\lheight{p} = \lheight{\hat{p}}$ and we are done.
    Else, by \cref{lem:modgheightbounds} and the bound $\deg_z(\hat{p}) \leq em-1$, we have 
    $$\lheight{p} = \lheight{\hat{p} \bmod q} \leq h \cdot em^2 + dem^2\log(n+2) + 2 em\log(em).   $$

    For item 3, consider the representation $\hat{g} \in \bZ[y_1, \dots, y_m][z]$ and the representation of $\widehat{f_i} \in R\bs{z}$.
    Let $p := \hat{g}(\widehat{f_1}, \dots, \widehat{f_m}) \in R[z]$.
    Thus, $\deg_z(p) \leq (e-1) \cdot (d_g + 1)$ and by \cref{lem:kpsheightbound} item 3, 
    $$\lheight{p} \leq \lheight{g} + d_g \cdot (h + \log(m+1) + \max\bc{e-1, d} \cdot \log(n+2))$$
    Since $g(f_1, \dots, f_m) = p \bmod q$, by \cref{lem:modgheightbounds} we have 
    \begin{align*}
        \lheight{g(f_1, \dots, f_m)} \leq 2ed_g \cdot \bs{\lheight{g} + d_g\cdot (h + \log(m+1) + \max\bc{e, d} \cdot \log(n+2)) + 2\log(2ed_g)}. & \qedhere
    \end{align*}
\end{proof}

\subsection{Height bounds for primitive elements}

In this section we gather height bounds for representations of primitive elements.
The following theorem is stated in \cite[Theorem 4]{K96}.
\RO{TODO for myself later: put the proof of the theorem below somewhere, which is simply the same proof of Koiran but using a version of \cite[Theorem 3]{K96} which keeps track of the complexity of the polynomials $Q_i$ - this can be done by looking at the proof of Canny, which simply fleshes out the proof of van der Waerden...

Will do this after the submission, since we don't hav time and this is easy (but tedious) to do.
}

\begin{theorem}[Complexity of primitive element]\label{theorem: koiran primitive element}
There is a universal constant $c \geq 1$ such that the following holds.
Let $\alpha_1, \alpha_2, \dots, \alpha_m$ be $m$ algebraic numbers which are roots of polynomials $P_i \in \bZ\bs{z}$ with $\deg(P_i) \leq d$ and $\lheight{P_i} \leq h$.
There is a primitive element $\gamma$ for $\alpha_1, \dots, \alpha_m$ which is a root of an irreducible polynomial $q \in \bZ\bs{z}$ with $\deg(q) \leq d^m$ and $\lheight{q} \leq h \cdot d^{m^c}$.
Moreover, there are non-zero integers $a_i$ and polynomials $Q_i \in \bZ\bs{z}$ with $\lheight{a_{i}} \leq h \cdot d^{m^c}$, $\deg(Q_i) < \deg(q)$ and $\lheight{Q_i} \leq h \cdot d^{m^c}$ such that $\alpha_i = Q_i(\gamma)/a_i$ for every $i \in [m]$.
\end{theorem}

We now state \cite[Theorem 7]{K96}.

\begin{theorem}\label{theorem: koiran solution bound}
    There is an absolute constant $c > 0$ such that if the system $f_1, \dots, f_m \in R$ has a solution over $\overline{\bQ}$, then there is a solution $\alpha = (\alpha_1, \ldots, \alpha_n)$ such that each $\alpha_i$ is a root of a polynomial $P_i \in \bZ[z]$ satisfying $\deg(P_i) \leq 2^{(n \log \sigma)^{c}}$ and $\lheight{P_i} \leq h \cdot 2^{(n \log \sigma)^{c}}$, where $\sigma := 2 + \sum_{i=1}^m \deg(f_i)$.
\end{theorem}

\subsection{Height bounds for membership in ideals}\label{subsection: height bounds ideals}

\RO{rephrase this a bit later}
We now state some height bounds on the polynomials occurring in the Nullstellensatz over $\bZ$ and $\bZ\bs{\alpha}$, where $\alpha$ is an algebraic integer.
Results of this nature are known as effective arithmetic Nullstellensatz, and the current best known bounds can be found in \cite{kps01}.
While the optimal bounds for polynomials with integer coefficients are easy to state, the optimal bounds for the more general case of polynomials with algebraic integers are significantly more technical.
Their statements involve extending the notion of logarithmic height to algebraic number fields, which is fairly intricate.
We choose to forego these optimal bounds in the interest of simplicity, as the more elementary (and weaker) bounds suffice for our purposes.

\paragraph{Subsection setup.} For the rest of this subsection, we let $q \in\bZ\bs{z}$ be a monic irreducible polynomial with $\deg{q} = e$ and $\lheight{q} \leq h$, and let $\alpha \in \overline{\bQ}$ be a root of $q$.
Let $A := \bZ\bs{\alpha}\bs{x_{1}, \dots, x_{n}}$.
By our convention, we represent polynomials in $A$ as elements of $R\bs{z}$ with $z$-degree less than $e$.
In the statements that follow, we will consider a degree parameter $d \geq 3$, and polynomials $f_1, \dots, f_m \in A$ with $\deg{f_i} \leq d$ and $\lheight{f_i} \leq h$.
We will denote by $I := (f_1, \dots, f_m) \subseteq S$ the ideal generated by $f_1, \dots, f_m$ and by $V := Z(I) \subseteq \overline{\bQ}^n$.

\begin{lemma}
    \label{lem:nullstellensatzheightbound}
    If $V(I) = \emptyset$, then there exists $a \in \bZ\bs{\alpha} \setminus \bc{0}$ and $g_{1}, \dots, g_{m} \in A$ satisfying
    \begin{equation}\label{equation: integral nullstellensatz equation height bound}
        a = f_{1} g_{1} + \cdots + f_{m} g_{m},
    \end{equation}
    with $\deg{g_{i}} \leq 2d^{n}$ and $\lheight{g_{i}}, \lheight{a} \leq h \cdot e^{5} \cdot d^{5n^{2}}$.
\end{lemma}
\begin{proof}
    Since $V(f_{1}, \dots, f_{m}) = \emptyset$, the effective Nullstellensatz \cref{theorem: jelonek nullstellensatz} shows that there are polynomials $h_{1}, \dots, h_{m} \in S$ with with $\deg{h_{i}} \leq 2 d^{n}$ such that 
    \begin{equation}\label{equation: nullstellensatz height linear system}
        1 = \sum f_{i} h_{i}.
    \end{equation}
    Let $M := M_{D}(f_{1}, \dots, f_{m})$ be the matrix of the linear system \cref{equation: nullstellensatz height linear system} (as defined in \cref{section: preliminaries}) with $D := 2d^{n}$. 
    Note that the entries of $M$ are elements of $\bZ[\alpha]$ (hence given by elements of $\bZ\bs{z}$ of degree less than $e$).
    \cref{theorem: jelonek nullstellensatz} implies that the linear system $1 = Mv$ has a solution with coefficients in $\overline{\bQ}$, and since the coefficients of $M$ lie in $\bQ\bs{\alpha}$, this implies the existence of a solution with coefficients in $\bQ\bs{\alpha}$.
    By \cref{lem:generalcramer}, there exists a solution where each unknown $v_{j}$ is either $0$, or of the form $\frac{det N_{j}}{\det N}$, where $\det N_{j}, \det N$ are minors of $\bs{M | 1}$.
    Note that the denominator is common among all the unknowns.
    We will take our polynomials $g_1, \dots, g_m$ to be the polynomials in $A$ corresponding to the vector $\det{N} \cdot v$ (note that this is a vector with entries in $A$).
    By the above, this choice gives us a solution to \cref{equation: integral nullstellensatz equation height bound} with $g_i \in A$ and $\deg{g_i} \leq 2 d^n$.
    We now need to bound $\lheight{g_i}$.

    Since $\lheight{g_i} \leq \max_j\{\lheight{\det{N_j}}\}$, it is enough to bound $\lheight{\det{N}}$ and $\lheight{\det{N_j}}$.
    Each square submatrix $N, N_j$ has entries in $\bZ[\alpha]$ and dimension upper bounded by $D' := \binom{D + d + n}{n} \leq d^{n^2 + n}$, which is the number of rows of $M$.
    As $\lheight{M_{ij}} \leq h$ for all $i,j$, \cref{corollary: determinant height bounds algebraic integer} yields 
    \begin{align*}
        \lheight{\det{N}}, \lheight{\det{N_j}} 
        &\leq eD'\br{D' \br{e + h + \log D'} + 2 \log\br{eD'}} \\
        &\leq ed^{n^2+n}\br{d^{n^2+n} \br{e + h + (n^2 + n) \cdot\log d} + 2 \log\br{ed^{n^2+n}}} \leq h \cdot e^5 \cdot d^{5 n^2} \qedhere
    \end{align*}
\end{proof}

The following easy corollary extends this result to arbitrary membership in radical ideals.

\begin{corollary}
    \label{cor:radicalmembership}
    Let $f \in A$ be such that $\deg{f} \leq d$ and $\lheight{f} \leq h$.
    If $f \in \radideal{I}$ then there exist $t \in \bN$, $a \in \bZ\bs{\alpha} \setminus \bc{0}$, and $g_{1}, \dots, g_{m} \in A$ satisfying
    $$ a f^{t} = f_{1} g_{1} + \cdots + f_{m} g_{m},$$
    with $t \leq 2(d+1)^{n+1}$, $\deg{g_{i}} \leq 2(d+1)^{n+2}$, and $\lheight{g_{i}}, \lheight{a} \leq 12 \cdot h \cdot e^5 \cdot (d+1)^{6(n+1)^2}$.
\end{corollary}

\begin{proof}
    Let $y$ be a new variable, and consider the polynomials $f_{1}, \dots, f_{m}, 1-yf \in A[y]$.
    As $f \in \radideal{f_{1}, \dots, f_{m}}$, we have $Z(f_1, \dots, f_m, 1 - yf) = \emptyset$, and thus \cref{lem:nullstellensatzheightbound} implies that there exist $a \in \bZ\bs{\alpha} \setminus \bc{0}$ and $h, h_{1}, \dots, h_{m} \in A$ such that $a = \sum f_{i} h_{i} + h(1 - yf)$.
    Moreover, \cref{lem:nullstellensatzheightbound} implies that $\deg(h_i), \deg(h) \leq D := 2(d+1)^{n+1}$. 
    Substituting $y = 1/f$ in the above equation and multiplying by $f^t$ to clear denominators, where $t \leq D$ (due to the degree bounds), we obtain equation 
    $af^{t} = \sum f_{i} g_{i}$, where we have $g_{i} = f^{t} \cdot h_{i}(x_{1}, \dots, x_{n}, 1/f)$.

    We now show the bounds on the heights and degrees.
    Let $h_i'(x_1, \dots, x_n, y) := y^{\deg_y(h_i)} \cdot h_i(x_1, \dots, x_n, 1/y)$.
    Since $h_i$ and $h_i'$ have the same coefficient set, we have $\lheight{h_i'} = \lheight{h_i}$.
    Moreover, $\deg{h_i'} = \deg(h_i)$.
    
    By \cref{lem:nullstellensatzheightbound} we also have that $\lheight{h_{i}}, \lheight{a} \leq h \cdot e^{5} \cdot (d+1)^{5(n+1)^{2}}$, and also $\deg{h_{i}} \leq D$.
    Since $g_i = f^{t-\deg_y(h_i)} \cdot h_i'(x_1, \dots, x_n, f)$, we have that $\deg(g_i) \leq t \cdot \deg(f) + \deg(h_i) \leq D \cdot (d+1)$.
    Moreover, \cref{corollary: height bounds polynomials algebraic number coefficients} items 2 and 3 imply the desired height bounds.
\end{proof}

\RO{we could also simplify here and merge the two statements above with the more general one, since the proofs are essentially the same and give (to us) the same height and degree upper bounds - let's do that after the deadline?}

We will use the above lemmas to show bounds on the defining equations for the image of $Z\ideal{f_{1}, \dots, f_{m}}$ under certain projections.
This is an arithmetic version of \cref{lem:dimensionelimination}, albeit with stronger assumptions.

\begin{lemma} \label{lem:eliminationinQ}
    Suppose $I$ is radical and $\dim I = 1$.
    Define $\pi: \bA^{n} \to \bA^{2}$ to be projection to the first two coordinates.
    If $\dim \pi(V) = 1$ and $\pi(V)$ is equidimensional, then $\pi(V) = Z(g)$ for some $g \in \bZ\bs{\alpha}\bs{x_{1}, x_{2}}$.
    Further, there exists $a \in \bZ\bs{\alpha} \setminus \bc{0}$, and $h_{1}, \dots, h_{m} \in A$ satisfying 
    $$a g = \sum f_{i} h_{i},$$
    with $\deg{h_{i}} \leq d^{4n^{2}}$ and $\lheight{g}, \lheight{h_{i}} \leq h \cdot e^{2} \cdot (d)^{12n^{3}}$.
\end{lemma}

\begin{proof}
    The ideal of $\pi(V)$ is exactly the elimination ideal $I \cap \overline{\bQ}\bs{x_{1}, x_{2}}$.
    By the assumption that $\pi(V)$ is equidimensional with $\dim \pi(V) = 1$, this ideal is principal.
    Let $G$ be the reduced \grobner{} basis of $I$ with respect to the lexicographic order induced by the variable order $x_{n} \succ x_{n-1} \succ \cdots \succ x_{1}$.
    Since $I$ is generated in $A$, we have $G \subset \bQ\bs{\alpha}\bs{x_{1}, \dots, x_{n}}$. 
    Thus, $I \cap \overline{\bQ}\bs{x_{1}, x_{2}}$ is generated by the elements of $G \cap \overline{\bQ}\bs{x_{1}, x_{2}}$. 
    Since $I \cap \overline{\bQ}\bs{x_{1}, x_{2}}$ is principal, we must have $G \cap \overline{\bQ}\bs{x_{1}, x_{2}} = \bc{\hat{g}}$. 
    Since $\hat{g} \in G$ we have $\hat{g} \in \bQ(\alpha)[x_1, x_2]$ and $\lterm{\hat{g}} = 1$.
    Let $b \in \bZ$ be the common denominator of $\hat{g}$ and $g := b \cdot \hat{g} \in \bZ[\alpha][x_1, x_2]$. 
    By the above, $\pi(V) = Z(g)$.

    By \bezout's theorem (\cref{theorem: bezout inequality}) we have $\deg{g} \leq d^{n}$.
    By the above paragraph, $g$ is the unique (up to scalar multiplication) lowest degree polynomial in $I \cap \overline{\bQ}\bs{x_{1}, x_{2}}$.
    Let $B := d^{4n^2}$. 
    Since $\dim(I) = 1$ and $B$ is larger than $\deg{g}$ and also larger than the representation bound from \cref{corollary: representation degree ideal membership}, there are forms $h_i \in \bQ\bs{\alpha}\bs{x_{1}, \dots, x_{n}}$ with $\deg{(f_{i} h_{i})} \leq B$ such that $g = \sum f_{i} h_{i}$.
    
    Let $M := M_{B}(f_{1}, \dots, f_{m})$ be the matrix corresponding to the linear system as described in \cref{section: preliminaries}.
    Let $M'$ be the matrix obtained by dropping the rows of $M$ that correspond to monomials in $x_{1}, x_{2}$ that are smaller than $\lm{g}$ in the monomial order.
    Any solution $\widehat{h_{1}}, \dots, \widehat{h_{m}}$ of the linear system $M'v = \lm{g}$ is such that $\sum f_{i} \widehat{h_{i}} \in I \cap \overline{\bQ}\bs{x_{1}, x_{2}}$, and $\lm{\sum f_{i} \widehat{h_{i}}} = \lm{g}$.
    Therefore, for any such solution $\widehat{h_{1}}, \dots, \widehat{h_{m}}$ we have $\sum f_{i} \widehat{h_{i}} = \frac{1}{\lterm{g}} \cdot g = b^{-1} \cdot g$.

    By Cramer's rule (\cref{lem:generalcramer}), there is a solution $\widetilde{h_{i}}$ whose coefficients are of the form $\det M_{j} / \det N$, where $M_{j}, N$ are submatrices of the augmented matrix $\bs{M' | \lm{g}}$.
    Let $B' := \binom{B + n}{n}$, that is, $B'$ is the number of monomials in the variables $x_1, \dots, x_n$ of degree at most $B$. 
    By our upper bound on $B$, we have $B' \leq \br{4d}^{4n^{3}}$.
    Note that $B'$ is an upper bound on the number of rows of the matrix $\bs{M' | \lm{g}}$.
    Thus, by \cref{corollary: determinant height bounds algebraic integer}, we have $\lheight{\det M_j}, \lheight{\det N} \leq e^2 \cdot h \cdot (4d)^{9n^3}$.
    
    Multiplying the equation $\sum f_{i} \widetilde{h_{i}} = \frac{1}{b} \cdot g$ throughout by $\det N$, we get that $\frac{\det N}{b} \cdot g = \br{\det N} \cdot \sum f_i \widetilde{h_i} \in \bZ[\alpha][x_1, x_2] \setminus \bc{0}$.
    Let $h_i := \det(N) \cdot \widetilde{h_i}$.
    As the coefficients of $h_i$ are given by the minors $\det(M_j)$, we have $\lheight{h_i} \leq \max_j \lheight{\det(M_j)} \leq e^2 \cdot h \cdot (4d)^{9n^3}$.
    Since $\frac{\det N}{b} \cdot g \in \bZ[\alpha][x_1, x_2]$, and by the definition of $g$, we must have $a := \frac{\det N}{b} \in \bZ[\alpha] \setminus \bc{0}$.
    In particular, this implies $b \mid \det N$ and therefore we have $\lheight{a}, \lheight{b} \leq \lheight{\det{N}} \leq e^2 \cdot h \cdot (4d)^{9n^3}.$
    As $\lheight{b^{-1} \cdot g} \leq e^2 \cdot h \cdot (4d)^{9n^3}$ and $b \in \bZ$, we have that $\lheight{g} \leq 2 \cdot e^2 \cdot h \cdot (4d)^{9n^3}$.
\end{proof}

We will need the following extension of the above lemma, where we drop the assumption that $I$ is radical.

\begin{lemma}
    \label{lem:eliminationextendedeffective}
    Suppose $\dim I = 1$.
    Define $\pi: \bA^{n} \to \bA^{2}$ to be projection to the first two coordinates.
    If $\dim \pi(V) = 1$ and $\pi(V)$ is equidimensional, then $\pi(V) = Z(g)$ for some $g \in \bZ\bs{\alpha}\bs{x_{1}, x_{2}}$.
    Moreover, there exist an absolute constant $c > 0$, $t \in \bN$, $a \in \bZ\bs{\alpha}$, and $h_{1}, \dots, h_{m} \in A$ such that $a g^{t} = \sum f_{i} h_{i}$, with $t \leq 4 d^{2n}$, $\deg{\br{f_i h_{i}}} \leq 5d^{3n}$, and $\lheight{a}, \lheight{h_{i}} \leq h \cdot e^{c} \cdot d^{cn^{c}}$.
\end{lemma}

\begin{proof}
    Since $\pi(V)$ is equidimensional and $\dim \pi(V) = 1$, we have that $I(\pi(V)) = (\hat{g})$ for some $\hat{g} \in \overline{\bQ}[x_1, x_2]$.
    Let $G$ be the reduced \grobner{} basis of $I$ with respect to the lexicographic order induced by the variable order $x_{n} \succ x_{n-1} \succ \cdots \succ x_{1}$.
    Since $I$ is generated in $A$, we have $G \subset \bQ\bs{\alpha}\bs{x_{1}, \dots, x_{n}}$. 
    Thus, $J := I \cap \overline{\bQ}\bs{x_{1}, x_{2}}$ is generated by the elements of $G \cap \overline{\bQ}\bs{x_{1}, x_{2}}$. 
    Since $\radideal{J} = \radideal{I} \cap \overline{\bQ}[x_1, x_2]$, by the closure theorem \cite[Chapter 3.2, Theorem 3]{cox1997ideals} we have that $\radideal{J} = (\hat{g})$.
    
    Now, let $J' := J \cap \bQ[\alpha]\bs{x_1, x_2}$ and let $\radideal{J'}$ be the radical ideal of $J'$ in $\bQ[\alpha]\bs{x_1, x_2}$. 
    Since $\dim J' = 1$, we have that $\radideal{J'} = (g')$, for some $g' \in \bQ[\alpha]\bs{x_1, x_2}$.
    Since $\bQ$ is a perfect field, by  \cite[\href{https://stacks.math.columbia.edu/tag/030U}{Tag 030U}]{stacks-project}, we know that $\radideal{J} = \radideal{J'} \otimes \overline{\bQ}$.
    Thus, we have that $\radideal{J} = (g') \overline{\bQ}\bs{x_1, x_2}$.
    Letting $g \in \bZ\bs{\alpha}\bs{x_1, x_2}$ be the polynomial obtained by clearing the common denominator of $g'$, we have $\radideal{J} = (g)$ as we wanted.

    By Bézout's theorem (\cref{theorem: bezout inequality}) we have $\deg g \leq d^{n}$.
    Further, by the effective nullstellensatz (\cref{theorem: jelonek nullstellensatz}) and the Rabinowitsch trick, there exists $t \in \bN$ and $h_{i}' \in \bQ\bs{\alpha}\bs{x_{1}, \dots, x_{m}}$ such that $g^{t} = \sum f_{i} h_{i}'$.
    Moreover, we have $t \leq 4d^{2n}$ and $\deg\br{f_i h_{i}'} \leq 5d^{3n}$.
    We assume that $t$ is the smallest power for which we can write $g^{t}$ with these degree bounds.

    Consider the linear system $M := M_{D}(f_{1}, \dots, f_{m})$ with $D = 5d^{3n}$.
    Let $M'$ be the system obtained by dropping those rows of $M$ corresponding to monomials that are smaller than $\lm{g}^{t} = \lm{g^{t}}$.
    Any solution to $M'v = \lm{g}^{t}$ is a monic polynomial in $\bQ\bs{\alpha}\bs{x_{1}, x_{2}}$ with leading coefficient $1$, and leading monomial $\lm{g}^{t}$.
    Let $h_{1}', \dots, h_{m}'$ be a solution to this system obtained by applying Cramer's rule (\cref{lem:generalcramer}). 
    Thus, there is a submatrix $N$ of $\bs{M' \mid \lm{g}^{t}}$ such that $\det{N}$ is the common denominator of every coefficient of $h_i'$ for $i \in [m]$.
    Moreover, there are submatrices $M_j$ of $\bs{M' \mid \lm{g}^{t}}$ such that each coefficient $\det{N} \cdot h_i'$ is given by some $\det{M_j}$.
    In particular, \cref{corollary: determinant height bounds algebraic integer} implies that $\lheight{\det{N} \cdot h_i'} \leq e \cdot D \cdot (2ehD + D \log D + 2 \log(eD)) \leq 4e^2hD^3 \leq h \cdot e^2 \cdot d^{10n}$.
    
    Let $g_{1} := \sum f_{i} \cdot \br{\det{N} \cdot h_{i}'}$.
    Hence, we have $g_{1} \in A$, and \cref{corollary: height bounds polynomials algebraic number coefficients} implies $\lheight{g_{1}} \leq h \cdot e^4 \cdot d^{14n}$.
    Let $\gamma \in \bZ\bs{\alpha}$ be the leading coefficient of $g_{1}$.
    Treating $\gamma$ as an element of $\bZ\bs{z}$ of degree at most $e$, we see that $\gamma$ and $q$ are relatively prime.
    By \cref{cor:resdischeightbounds}, we can find an element $\delta \in \bZ\bs{z}$ with $\deg(\delta) < e $, and $b := \resultant{\gamma}{q}{z} \in \bN^*$ such that $\gamma \delta \equiv b \bmod q$.
    Moreover, \cref{cor:resdischeightbounds} implies $\lheight{\delta}, \lheight{b} \leq h \cdot e^4 \cdot d^{18n}$.
    Replacing $g_{1}$ by $\delta g_{1}$, we can assume that the leading coefficient of $g_{1}$ is $b$, and $\lheight{g_{1}} \leq h \cdot e^6 \cdot d^{22n}$.

    Since $g_{1} \in I \subset \radideal{I} = \ideal{g}$, we have $g \mid g_{1}$ in $\bQ\bs{\alpha}\bs{x_{1}, x_{2}}$. 
    Let us write $g_{1} = g g_{2}$.
    We use this to deduce a bound on $\lheight{g}$.
    Let $O_{L}$ be the ring of integers of $\bQ\bs{\alpha}$. 
    Thus, we have $\bZ\bs{\alpha} \subset O_{L} \subset \bZ\bs{\alpha}_{\Delta}$, where $\Delta := \disc{z}{q}$.
    In the ring $\br{O_{L}}_{b}\bs{x_{1}, x_{2}}$, the polynomial $g_{1}$ is monic.
    By Gauss's lemma \cref{lem:gauss's lemma} we can assume that $g, g_{2} \in \br{O_{L}}_{b}\bs{x_{1}, x_{2}}$.
    Further, after rescaling, we can assume that $g', g'_{2}$ are the monic forms of $g, g_{2}$ in $\br{O_{L}}_{b}\bs{x_{1}, x_{2}}$.
    By \cite[pp. 64-67]{lenstra}, there is a universal constant $c_1 > 0$ such that every coefficient of $g'$, when written as a monic element of $\bZ\bs{\alpha}_{Db}$, uses rational numbers of absolute value $\leq 2^{h \cdot e^{c_{1}} \cdot d^{c_{1}n^{c_{1}}}}$.
    Multiplying by $\Delta b$ shows that there is a universal constant $c_2 > 0$ such that $\lheight{g} \leq h \cdot e^{c_{2}} \cdot d^{c_{2}n^{c_{2}}}$.

    Now that we have control on $\lheight{g}$, by Cramer's rule applied to $M v = g^t$, there are $a \in \bZ\bs{\alpha}$ and $h_{1}, \dots, h_{m} \in A$ such that $a g^{t} = \sum f_{i} h_{i}$, with $\deg{f_i h_{i}} \leq D$, and $\lheight{a}, \lheight{h_{i}} \leq h \cdot e^{c} \cdot d^{cn^{c}}$ for some universal constant $c> 0$.
    This height bound comes from noticing that, by Cramer's rule, $a$ and every coefficient of $h_i$ is a minor of the matrix $\bs{M \mid g^t}$, and we have upper bounds on $\lheight{M}$ and $\lheight{g^t}$.
\end{proof}

\AG{think about the statements for higher dimension}

\subsection{Absolutely irreducible factors of bivariate polynomials}

In \cite{kaltofen1995effective}, Kaltofen gives an algorithm to factor bivariate polynomials with coefficients in a field $\bF$ over the algebraic closure $\overline{\bF}$.
The algorithm uses the usual template of factoring algorithms: given $f(x, y)$, the algorithm begins by picking a root $\alpha \in \overline{\bF}$ of $f(x, 0)$ in $\overline{\bF}$, which is the same as the linear factor $x - \alpha$.
This factorisation is then lifted, using Newton iteration, modulo $y^{t}$ for increasing powers of $t$.
After some steps of lifting, a reconstruction step is performed and a factor $f_{1}$ of $f$, over $\overline{\bF}$, is obtained.

A key observation made in \cite{kaltofen1995effective} is that the lifting and reconstruction steps only use arithmetic in the field $\bF$, even though the factor $(x - \alpha)$ that the Newton iteration begins with has coefficients in $\overline{\bF}$.
In particular this implies that the factor $f_{1}$ has coefficients in $\bF\br{\alpha}$.
Since the minimal polynomial of $\alpha$ is a factor of $f(x, 0)$, the result can be interpreted as the fact that factors of $f$ can be found in low complexity extensions of $\bF$.
Further, every irreducible factor of $f$ can be obtained this way, therefore every irreducible factor of $f$ lies in some low complexity extension of $\bF$.
While each factor lies in a low complexity extension of $\bF$, since the number of factors of $f$ can be large, the compositum of all these extensions might be very large.
This is a common phenomenon in computational algebraic geometry, see \cite[Section~4]{bayer1993can}.
Following the principle, we will only factor as much as is required, in order to ensure that we can work with low complexity extensions.

The next lemmas are from \cite{kaltofen1995effective}, and they formalise the discussion in the previous two paragraphs.

\begin{lemma}
    \label{lem:kaltofen factoring algorithm}
    Given polynomial $f \in \bF\bs{x, y}$ that monic in $x$, such that $\deg{f} = d$, and such that the resultant $\resultant{f(x, 0)}{\partial f(x, 0) / \partial x}{x} \neq 0$, there exists an algorithm that returns a list of absolutely irreducible factors $f_{1}, \dots, f_{r}$ of $f$, with each $f_{i} \in \bF(\alpha_{i})\bs{x, y}$, where $\alpha_{i}$ is a root of $f(x, 0)$.
\end{lemma}
While the result gives an actual algorithm that computes the factors, we are just interested in the fact that every factor of $f$ is of this form.
The list of factors returned by the algorithm is not necessarily distinct, and computing the set of distinct factors is a non-trivial task depending on the model used to represent the extensions $\bF(\alpha_{i})$.
We now state bounds on the logarithmic heights of the factors.
The following is \cite[Corollary~1, Theorem~3]{kaltofen1995effective}, although we state a simplified statement specialised to our setting.
\begin{lemma}
    \label{lem:factorheight}
    Suppose $f \in \bZ\bs{x, y}$ is a squarefree polynomial with $\deg(f) = d$ and $\lheight{f} = h$ that is monic in $x$, and such that $f(x, 0)$ is squarefree.
    Suppose $f_{1}$ is an absolutely irreducible factor of $f$.
    
    There exists $g(x)$ an irreducible (in $\bZ\bs{x}$) factor of $f(x, 0)$, and polynomials $\widetilde{f}_{1}, f_{2}, h \in \bZ\bs{x, y, z}$ and integers $\Delta_{1}, \Delta_{2}$ with $\lheight{\Delta_{1}}, \lheight{\Delta_{2}}, \lheight{\widetilde{f}_{1}} = O(h d^{c})$ such that $\Delta_{1} \Delta_{2} f = \widetilde{f}_{1} f_{2} + g(z)h$.
    Further, there exists a root $\alpha$ of $g$ such that $f_{1}(x, y) = \widetilde{f}_{1}(x, y, \alpha)$.
\end{lemma}

\begin{proof}
    The algorithm of \cref{lem:kaltofen factoring algorithm} shows the existence of $g, \alpha$, and also monic polynomials $\widehat{f}_{1}, \widehat{f_{2}} \in \bQ\bs{x, y, z}$ such that $f = \widehat{f}_{1}\widehat{f}_{2} + g(z)h$ for some $h \in \bQ\bs{x, y, z}$.
    By \cite[Corollary~1]{kaltofen1995effective}, the coefficients of $\widehat{f}_{1}$ are the solutions of a certain linear system, and by \cite[Theorem~3]{kaltofen1995effective} the minors of this system have logarithmic height $\bigO{h \cdot d^{c_{1}}}$.
    Using \cref{cor:detheightbound} to invert the denominators, we can write the coefficients of $\widehat{f}_{1}$ as elements of $(1/\Delta_{1}) \bZ\bs{z}$, where $\lheight{\Delta_{1}} = \bigO{h \cdot d^{c_{2}}}$.
    Further, the coefficients themselves also have logarithmic height $\bigO{h \cdot d^{c_{2}}}$.

    No similar bounds are provided on the logarithmic height of $\widehat{f}_{2}$, however since $f, \widehat{f}_{1}, \widehat{f}_{2}$ are monic we can use the fact that the coefficients of $\widehat{f}_{2}$ are algebraic integers, in particular they are contained in $(1/\Delta_{2}) \bZ\bs{z}$ where $\Delta_{2} = \disc{g}{z}$.
    Writing $\widehat{f}_{2}$ in this form and clearing denominators from the equation $f = \widehat{f}_{1}\widehat{f}_{2} + g(z)h$ gives us the required equation.
\end{proof}

Recall the notation of the number theory preliminaries.
As a corollary of the factoring algorithm, \cite[Theorem~8]{kaltofen1995effective} obtains a bound on primes $p$ for which an absolutely irreducible polynomial remains absolutely irreducible after applying the map $\bZ \to O_{\bF} / \kp$ for some prime $\kp$ above $p$.
\begin{theorem}
    \label{thm:kaltofen mod p irreducible}
    Let $A = \bZ\bs{\alpha}\bs{x_{1}, \dots, x_{m}}$ for an algebraic integer $\alpha$ with minimal polynomial $q$ satisfying $\deg{q} = e$ and $\lheight{q} \leq h$.
    Suppose $f \in A$ is absolutely irreducible with $\lheight{f} \leq h$ and $\deg{f} \leq d$.
    There exists an integer $\Delta$ with $\lheight{\Delta} \leq c \br{e d^{6} \log\br{h} + e d^{6}n \log\br{d}}$ for a universal constant $c$, such that for any prime $p \not | \Delta$ and any prime $\kp$ above $p$, the image of $f$ under the map $A \to O_{\bF} / \kp$ remains absolutely irreducible.
    
\end{theorem}

%% file: base-change-alg-int.tex
In this section, we prove our technical base change results.
We first show that certain mod $p$ reductions preserve dimensions of zerosets \cref{lem:dimpreserve}.
We then show that irreducibility of dimension $0$ and dimension $1$ varieties is also preserved under such base changes.
Finally, we show that reducible zerosets of dimension $0$ and $1$ remain reducible under base changes, and further, for sufficiently many primes, there are irreducible components of the zeroset after base change that are definable over $\bF_{p}$ itself.
This last fact will be crucial when we apply the Lang-Weil bounds \cref{thm:effective lang weil} in a subsequent section.
We split the study of irreducibility and reducibility into two parts, based on the dimension.
The dimension $0$ case (\cref{cor:onerootheightbound} and \cref{lem:tworootsbounds}) is relatively straightforward, and uses ideas from \cite{K96}
The dimension $1$ case (\cref{thm:irredpreserve} and \cref{thm:redpreserve}) is significantly more involved.

We show \cref{cor:onerootheightbound},  \cref{lem:tworootsbounds}, and \cref{thm:redpreserve} for zerosets that are defined over $\bZ$.
Our proof of \cref{thm:redpreserve} invokes \cref{lem:dimpreserve} and \cref{thm:irredpreserve} for zerosets that are defined not just over $\bZ$, but also over small algebraic extensions of $\bQ$.
Therefore, we prove \cref{lem:dimpreserve} and \cref{thm:irredpreserve} in this more general setting.

With all of the above in mind, we now set up the notation for this section.
We recall the preliminaries discussed in \cref{section: preliminaries}, especially the number theory preliminaries in \cref{subsec:number theory prelims}.
Let $q \in \bZ\bs{z}$ be a monic irreducible polynomial with $\deg{q} = e$ and $\lheight{q} \leq h$, and let $\alpha$ be a root of $q$.
Let $\bF := \bQ\bs{\alpha}$, and let $O_{\bF}$ be the ring of integers of $\bF$, we have $\bZ\bs{\alpha} \subset O_{\bF}$.
Recall that for each prime $\kp$ above $p$ with $p \not | \disc{z}{q}$, we have a quotient map $\bZ\bs{\alpha} \to O_{\bF} / \kp \cong \bF_{p}\bs{z} / q_{1}$, where $q_{1}$ is an irreducible factor of $q \pmod p$.

We define $A := \bZ\bs{\alpha}\bs{x_{1}, \dots, x_{n}}$.
Suppose $f_{1}, \dots, f_{m} \in A$ are polynomials with $\deg{f_{i}} \leq d$ with $d \geq 3$ and $\lheight{f_{i}} \leq h$.
Let $I := \ideal{f_{1}, \dots, f_{m}}$ and $V := Z(I)$.
For a prime $p \in \bN$ such that $p \not| \disc{z}{q}$, and for any prime $\kp \subset O_{\bF}$ that lies above $p$ corresponding to a factor $q_{1}$ of $q \pmod p$, let $I_{\kp}$ be the ideal of $(\bF_{p}\bs{z} / q_{1})\bs{x_{1}, \dots, x_{n}}$ generated by the images of $f_{1}, \dots, f_{m}$ under the map $\bZ\bs{\alpha} \to O_{\bF} / \kp$ described above.
Let $V_{\kp}$ denote the algebraic subset of $I_{\kp}$ in $\overline{\bF}_{p}\bs{x_{1}, \dots, x_{n}}$.
Define $\sigma := dm + 2$.

\begin{lemma}
    \label{lem:dimpreserve}
    There exists a universal constant $c$ and an integer $\Delta$ with $\lheight{\Delta} \leq h \cdot e^{c} \cdot 2^{\br{n \log \sigma}^{c}}$ such that for any prime $p \not | \Delta$ and for every $\kp$ above $p$ we have $\dim(V) = \dim(V_{\kp})$.
\end{lemma}

\begin{proof}
    Let $r$ be the dimension of $V$.
    We first control the primes $p$ for which $\dim V_{\kp} < r$, and then control the primes $p$ for which $\dim V_{\kp} > r$.
    
    By \cref{lem:dimensionelimination}, there is a subset of $r$ variables such that $I$ does not contain any polynomial supported on these $r$ variables.
    Without loss of generality, suppose these variables are $x_{1}, \dots, x_{r}$.
    Set $D := 4d^{3n}$, and $D_{r} := r D$.
    Set $D_{b} := D_{r} + 5 d^{3n}$.
    Let $g := x_{1}^{D} x_{2}^{D} \cdots x_{r}^{D}$, so $\deg{g} = D_{r}$.
    Let $M := M_{D_{b}}(f_{1}, \dots, f_{m})$ be the linear system corresponding to membership in $I$ as defined in \cref{section: preliminaries}.
    Let $M'$ be the system obtained by dropping all the rows of $M$ corresponding to monomials in $x_{1}, \dots, x_{r}$ that are smaller than $g$ in the graded lexicographic order.
    If $h_{1}, \dots, h_{m}$ is a solution to $M'v = g$, then $\sum f_{i} h_{i} \in \bQ\bs{x_{1}, \dots, x_{r}}$, with $\lm{\sum f_{i} h_{i}} = g$.

    The system $M'v = g$ is unsatisfiable since $I$ does not contain any polynomial supported on these $r$ variables.
    Therefore $\Rank{M'} < \Rank{\bs{M' | g}}$.
    Let $\delta_{0} \in \bZ\bs{\alpha}$ be any nonzero minor of $\bs{M' | g}$ of size $\Rank{\bs{M' | g}}$, the size of the submatrix corresponding to $\delta_{0}$ is at most $1 + \binom{D_{b} + D_{r} + n}{n}$.
    By \cref{cor:detheightbound} and \cref{lem:modgheightbounds} we can deduce that $\lheight{\delta_{0}} \leq h \cdot e^{c_{1}} \cdot 2^{\br{n \log \sigma}^{c_{1}}}$.
    Observe that $\delta_{0}$ and $q$ are coprime as elements of $\bZ\bs{z}$, since $q$ is irreducible.
    By \cref{cor:resdischeightbounds} we can deduce that $\Delta_{0} := \disc{z}{q} \cdot \resultant{q}{\delta_{0}}{z}$ satisfies $\lheight{\Delta_{0}} \leq h \cdot e^{c_{2}} \cdot 2^{\br{n \log \sigma}^{c_{2}}}$.

    We show that for every $p \not | \Delta_{0}$, and for every $\kp$ above $p$ we have $\dim V_{\kp} \geq r$.
    Let $J_{\kp} := I_{\kp} \overline{\bF}_{p} \bs{x_{1}, \dots, x_{n}}$.
    Suppose towards contradiction that $V_{\kp}$ has dimension less than $r$.
    By \cref{lem:dimensionelimination}, the ideal of $V_{\kp}$, that is $\radideal{J_{\kp}}$, contains a polynomial supported on $x_{1}, \dots, x_{r}$.
    By \cref{theorem: bezout inequality}, the degree of $V_{\kp}$ is bounded by $d^{n}$, therefore the degree the projection of $V_{p}$ to the coordinates $x_{1}, \dots, x_{r}$ is bounded by $d^{n}$.
    By \cite[Corollary~1.9]{scheiblechnerthesis}, the ideal $\radideal{J_{p}}$ contains a polynomial $f$ in the variables $x_{1}, \dots, x_{r}$ of degree at most $d^{n}$.
    By \cref{theorem: jelonek nullstellensatz} and the Rabinowitsch trick, there exist $e \leq 4 d^{2n}$ and $g_{i}$ with  $\deg{g_{i}} \leq 5 d^{3n}$ such that $f^{e} = \sum f_{i} g_{i}$.
    The degree of $f^{e}$ is at most $D = 4d^{3n}$.
    By potentially multiplying the equation throughout by a monomial, we can assume that there is a polynomial $f'$ with leading monomial $x_{1}^{D} x_{2}^{D} \cdot x_{r}^{D}$ such that $f' = \sum f_{i} g'_{i}$, where each $g'_{i}$ has degree at most $5 d^{3n} + r D$, which is exactly $D_{b}$.

    Observe that the matrix of the linear system corresponding to membership in $I_{\kp}$ is exactly the matrix $M$, where we apply the map $\bZ\bs{\alpha} \to \bF_{p}\bs{z} / q_{1}$ to each coefficient.
    Denote this by matrix by $M_{\kp}$.
    If we similarly drop the rows of $M_{\kp}$ corresponding to monomials in $x_{1}, \dots, x_{r}$ that are larger than $g$ in the graded lexicographic order, the matrix obtained is exactly the reduction of $M'$ under the same map, denote this matrix by $M'_{\kp}$.
    We have $\Rank{M'_{\kp}} \leq \Rank{M'}$, since any zero minor of $M'$ remains zero after reduction.
    Since $p \not | \Delta_{0}$, in particular $p \not | \resultant{q}{\delta_{0}}{z}$.
    Therefore, the polynomials $\delta_{0}$ and $q$ remain relatively prime in $\bF_{p}\bs{z}$.
    In particular, the minor $\delta_{0}$ remains nonzero under the map $\bZ\bs{\alpha} \to \bF_{p}\bs{z} / q_{1}$.
    The augmented matrix $\bs{M'_{\kp} | g}$ satisfies $\Rank{\bs{M'_{\kp} | g}} = \Rank{\bs{M_{\kp} | g}} = \Rank{M'} + 1$.
    In particular this implies that $M'_{\kp}v = g$ is unsolvable, which is in contradiction to the existence of $f'$.
    This completes the first part of the proof, that is controlling the primes $p$ for which $\dim V_{\kp} < r$.

    We now move on to the second part of the proof, that is, controlling the primes $p$ for which $\dim V_{\kp} > r$.
    Similar to the proof above, since $V$ has dimension $r$, for every subset of $r+1$ variables, say for example $x_{1}, \dots, x_{r+1}$, the ideal $I$ contains a polynomial $f'$ in $x_{1},\dots, x_{r+1}$ with leading monomial $g := x_{1}^{D} x_{2}^{D} \cdots x_{r+1}^{D}$.
    Further, we can write $f' = \sum f_{i}g_{i}$ with $\deg{g_{i}} \leq 5 d^{3n} + (r+1) D$.
    Set $D_{b} := 5 d^{3n} + (r+1) D$.
    If we consider the linear system $M_{D_{b}}$, and drop the rows corresponding to monomials in $x_{1}, \dots, x_{r+1}$ that are smaller than $x_{1}^{D} x_{2}^{D} \cdots x_{r+1}^{D}$ and call this new matrix $M'$, then $M'v = g$ is now satisfiable.
    In other words, $\Rank{M'} = \Rank{\bs{M' | g}}$.
    By \cref{lem:generalcramer}, there is a solution to this system, where the entries of $v$ have common denominator a minor of $\bs{M' | g}$.
    Call this minor $\delta_{\bs{r+1}}$, by \cref{cor:detheightbound} and \cref{lem:modgheightbounds} we can deduce that $\lheight{\delta_{\bs{r+1}}} \leq h \cdot e^{c_{3}} \cdot 2^{\br{n \log \sigma}^{c_{3}}}$.
    Set $\Delta_{\bs{r+1}} := \resultant{q}{\delta_{\bs{r+1}}}{z}$, for any $p \not | \Delta_{\bs{r+1}}$ and $p \not | \disc{q}{z}$, the image of this denominator is nonzero in $\bF_{p}\bs{z} / q_{1}$, and the ideal $I_{\kp}$ contains a polynomial in $x_{1}, \dots, x_{r+1}$.

    Similarly, we can construct $\delta_{U}, \Delta_{U}$ for every $U \subset \bs{n}$ with $\abs{U} = r+1$.
    If we let $\Delta_{1} := \disc{q}{z} \prod \Delta_{U}$, then for every prime $p \not | \Delta_{1}$ we must have $\dim V_{\kp} \leq r$.
    We have $\lheight{\Delta_{1}} \leq h \cdot e^{c_{4}} \cdot 2^{\br{n \log \sigma}^{c_{4}}}$.
    This controls all primes for which $\dim V_{\kp} > r$.
    Defining $\Delta := \Delta_{1} \cdot \Delta_{2}$ completes the proof.
\end{proof}

\subsection{Dimension zero}

Note that when $\dim I = 0$, the zeroset $V(I)$ is a finite set of points.
In this subsection, we assume that $f_{1}, \dots, f_{m} \in R$.

We first show that if $V$ consists of a single point, then the logarithmic height of this point is not too large.
\begin{lemma}
    \label{lem:onerootbound}
    Suppose $V = \{(\alpha_{1}, \dots, \alpha_{n})\}$.
    Then $\alpha_{i} \in \bQ$ and there is a universal constant $c > 0$ such that $\lheight{\alpha_{i}} \leq h \cdot 2^{\br{n \log \sigma}^{c}}$ for all $i \in [n]$.
\end{lemma}

\begin{proof}
    Let $K$ be the smallest Galois extension of $\bQ$ that contains $\alpha_{1}, \dots, \alpha_{n}$.
    Suppose $K \neq \bQ$, and without loss of generality assume $\alpha_{1} \not \in \bQ$.
    Then, there exists some element $\tau$ of the Galois group of $K / \bQ$ that is not identity on $\alpha_{1}$.
    Since $\tau$ fixes $f_i$, for $i \in \bs{m}$, as $f_i \in R$, we must have  $\br{\tau(\alpha_{1}), \cdots, \tau(\alpha_{n})} \in V$, contradicting the fact that $V = \{(\alpha_1, \dots, \alpha_n)\}$.

    We now obtain the bound on $\lheight{\alpha_{i}}$. 
    By \cref{theorem: koiran solution bound}, there is a universal constant $c_1 > 0$ and a polynomial $P_{i}$ with $\deg(P_i) \leq 2^{\br{n \log \sigma}^{c_{1}}}$, and $\lheight{P_i} \leq h \cdot 2^{\br{n \log \sigma}^{c_{1}}}$ such that $\alpha_{i}$ is a root of $P_{i}$.
    By \cref{lem:intunivarrationalroot}, there is a universal constant $c > 0$ such that $\lheight{\alpha_{i}} \leq h \cdot 2^{\br{n \log\sigma}^{c}}$.
\end{proof}

With the above lemma at hand, we now prove that if $|V| = 1$, then for all but finitely many primes $p$, we must have that $V_{p}$ consists of exactly one point, and this point lies in $\bF_{p}^{n}$.

\begin{corollary}
    \label{cor:onerootheightbound}
    Suppose $V = \{(\alpha_{1}, \dots, \alpha_{n})\}$.
    There is a universal constant $c$ and $\Delta \in \bZ$ with $\lheight{\Delta} \leq h \cdot 2^{\br{n \log \sigma}^{c}}$, such that for any prime $p$ with $\Delta \not\in (p) $, $V_p \subset \bF_p^n$ and $|V_p| = 1$.
\end{corollary}

\begin{proof}
    By \cref{lem:onerootbound}, there is a universal constant $c_1$ such that $\lheight{\alpha_{i}} \leq h \cdot 2^{\br{n \log \sigma}^{c_{1}}}$, for all $i \in [n]$.
    Let $\alpha_{i} = \beta_{i} / \gamma_{i}$ with $\beta_{i}, \gamma_{i} \in \bZ$ in reduced form, and let $g_{i} := \gamma_{i} x_{i} - \beta_{i}$.
    By the Nullstellensatz $g_{i} \in I$ for all $i \in [n]$.

    By \cref{cor:radicalmembership}, there are: a universal constant $c_2 > 0$, positive integers $a_{i}, e_{i}$ and polynomials $h_{ij} \in R$ such that $a_{i} g_{i}^{e_{i}} = \sum_{j} f_{j} h_{ij}$, where $e_i \leq 4(n+1)(d+1)^{n+1}$, $\deg(h_{ij}) \leq 8(n+1)(d+1)^{n+2}$ and lastly we have $\lheight{a_{i}}, \lheight{h_{ij}} \leq h \cdot 2^{\br{n \log \sigma}^{c_{1}}} \br{d+1}^{n} \br{nd\log{m}}^{c_{2}}$.

    Define $\Delta := \prod_{i} a_{i} \gamma_{i}$.
    For any prime $p$ such that $\Delta \not\in (p)$, the equation $a_{i} g_{i}^{e_{i}} = \sum_{j} f_{j} h_{ij}$ reduced mod $p$ implies that $g_i \in \radideal{I_p}$ over $\bF_p[x_1, \ldots, x_n]$.
    Therefore, $V_{p} = \{\br{\beta_1 \gamma_1^{-1}, \dots, \beta_n \gamma_n^{-1}}\}$.
    From the bounds on $\lheight{a_i}, \lheight{\gamma_i}$ we have $\lheight{\Delta} \leq (n+1) \cdot h \cdot 2^{\br{n \log \sigma}^{c_{1}}} \br{d+1}^{n} \br{nd\log{m}}^{c_{2}}$, which can be bounded by $h \cdot 2^{\br{n \log \sigma}^{c}}$ for an appropriate universal constant $c > c_2$.
\end{proof}

We now show that if $|V| \geq 2$, then the density of primes $p$ such that the set $V_p$ has at least two points in $\bF_p^n$ is large enough.
Our argument is an extension of the main argument of \cite{K96}.

\begin{lemma}
    \label{lem:tworootsbounds}
    Suppose $|V| \geq 2$.
    There exist: a universal constant $c$, a polynomial $G \in \bZ\bs{z}$ with $\deg(G) \leq 2^{\br{n \log \sigma}^{c}}$ and $\lheight{G} \leq h \cdot 2^{\br{n \log \sigma}^{c}}$, and $\Delta \in \bZ$ with $\lheight{\Delta} \leq h \cdot 2^{\br{n \log \sigma}^{c}}$, such that for any prime $p$ satisfying $\Delta \not\in (p)$ and $G \bmod p$ has a root in $\bF_{p}$, we have $|V_{p} \cap \bF_p^n | \geq 2$.
\end{lemma}

\begin{proof}
    Let $\br{\alpha_{1}, \dots, \alpha_{n}}$ and $\br{\beta_{1}, \dots, \beta_{n}}$ be two points in $V$.
    By \cite[Lemma~2]{K96}, there are a universal constant $c_1 > 0$ and polynomials $P_{ij} \in \bZ[x_1, \dots, x_n]$, where $i \in [2], j \in [n]$, with $\deg(P_{ij}) \leq 2^{\br{n \log \sigma}^{c_{1}}}$ and $\lheight{P_{ij}} \leq h \cdot 2^{\br{n \log \sigma}^{c_{1}}}$ such that $\alpha_{j}$ is a root of $P_{1j}$ and $\beta_{j}$ is a root of $P_{2j}$.
    \footnote{Note than while Theorem~7 in \cite{K96} is stated in terms of a single root, Lemma~2 of \cite{K96} holds for every root.}

    Let $c_2 \geq 1$ be the universal constant from \cref{theorem: koiran primitive element}.
    Applying \cref{theorem: koiran primitive element} to the $2n$ elements $\alpha_{1}, \dots, \alpha_{n}, \beta_{1}, \dots, \beta_{n}$ with the polynomials $P_{ij}$ gives us a primitive element $\gamma$, with minimal polynomial $G$, where $\deg(G) \leq 2^{\br{n \log \sigma}^{c_{2}}}$ and $\lheight{G} \leq h \cdot 2^{\br{n \log \sigma}^{c_{2}}}$.
    Further, we also obtain $Q_{ij} \in \bZ\bs{z}$ with $\deg(Q_{ij}) < \deg(G)$ and $\lheight{Q_{ij}} \leq h \cdot 2^{\br{n \log \sigma}^{c_{2}}}$, and $a_{ij} \in \bZ^*$ with $\log(|a_{ij}|) \leq h \cdot 2^{\br{n \log \sigma}^{c_{2}}}$ such that $\alpha_{j} = Q_{1j}(\gamma) / a_{1j}$ and $\beta_{j} = Q_{2j}(\gamma) / a_{2j}$.

    Let $p$ be a prime such that $G \bmod p$ has a root $x_{0}$ in $\bF_{p}$ and that $\prod_{i,j} a_{ij} \not\in (p)$.
    Then, the points $\br{\br{Q_{11}(x_{0}} / a_{11}, \dots, Q_{1n}(x_{0}) / a_{1n}}$ and $\br{\br{Q_{21}(x_{0}} / a_{21}, \dots, Q_{2n}(x_{0}) / a_{2n}}$ are in $V_{p}$.
    To complete the proof, it suffices to ensure that these two points differ on at least one coordinate, so we can claim that $V_{p}$ has at least two distinct points with coefficients in $\bF_p$.
    To this end, assume without loss of generality that $\alpha_{1} \neq \beta_{1}$, and therefore $Q_{11}/a_{11} \neq Q_{21}/a_{21}$.
    Consider the resultant $\resultant{a_{21}Q_{11} - a_{11}Q_{21}}{G}{z}$.
    Since $G$ is irreducible and $\deg(a_{21}Q_{11} - a_{11}Q_{21})\leq \max(\deg(Q_{11}), \deg(Q_{21})) < \deg(G)$, we can deduce $\resultant{Q_{11} - Q_{21}}{G}{z} \in \bZ^*$.
    Further, $\lheight{G}, \lheight{a_{21}Q_{11} - a_{11}Q_{21}} \leq h \cdot 2^{\br{n \log \sigma}^{c_{2}}}$.
    Since $\resultant{Q_{11} - Q_{21}}{G}{z} \in \bZ$ is the determinant of a matrix in the coefficients of $G, Q_{11} - Q_{21}$, it has logarithmic height at most $h \cdot 2^{\br{n \log \sigma}^{c_{3}}}$.
    
    Define $\Delta := a \cdot \resultant{Q_{11} - Q_{21}}{G}{z}$, so $\Delta$ has logarithmic height at most $h \cdot 2^{\br{n \log \sigma}^{c}}$.
    Suppose $p$ is a prime such that $p \not| \Delta$ and such that $G$ has a root in $\bF_{p}$.
    Since $p \not |\resultant{Q_{11} - Q_{21}}{G}{z}$, the polynomials $Q_{11} - Q_{21}$ and $G$ remain relatively prime in $\bF_{p}\bs{z}$.
    Therefore, for any root $x_{0}$ of $G$, we must have $Q_{11}(x_{0}) \neq Q_{21}(x_{0})$.
    Combined with the arguments above, this shows that for any such $p$, the algebraic set $V_{p}$ has at least two distinct points with coefficients in $\bF_{p}$.
\end{proof}

\subsection{Dimension one}

In the following theorem, we assume once again the more general setting where $f_{1}, \dots, f_{m} \in A$.

\begin{theorem}
    \label{thm:irredpreserve}
    Suppose $V$ is an irreducible curve.
    There exists an integer $\Delta$ with $\lheight{\Delta} \leq h \cdot e^{c} \cdot 2^{\br{n \log \sigma}^{c}}$ such that for any prime $p \not | \Delta$ and for every $\kp$ above $p$ the zeroset $V_{\kp}$ is irreducible.
\end{theorem}

\begin{proof}
    Let $D := d^{n}$ so $\deg(V) \leq D$ by \cref{theorem: bezout inequality}.
    Define $B := D^{2}$.
    Let $\Phi$ be the set of all linear maps $\bA^{n} \to \bA$, where each coefficient is picked from the set $\bs{B}$.
    The total number of such linear maps is $B^{n+1}$.
    For maps $\phi_{i}, \phi_{j} \in \Phi$, define $\phi_{ij}: \bA^{n} \to \bA^{2}$ to be the map with coordinate functions $\phi_{i}, \phi_{j}$.

    Let $\overline{V}$ be the projective closure of $V$, and let $V^{P} := \overline{V} \cap \bc{x_{0} = 0}$ be the intersection of $V$ with the hyperplane at infinity.
    Since $\deg(V) \leq D$, and since $V$ is defined by polynomials in $R$, the zeroset $V^{P}$ is finite and consists of at most $D$ points.
    By \cite[Theorem~1.15]{shafarevich1994basic}, for any $\phi_{i}$ such that the projective closure of the hyperplane $\phi_{i} = 0$ does not contain a point in $V^{p}$, the map $\phi_{i} : V \to \bA^{1}$ is a finite map.
    Define $\Phi' \subset \Phi$ to be the set of maps that are finite, the above argument shows that $\abs{\Phi'} \geq (1 - D/B) \abs{\Phi} = (1 - 1/D) \abs{\Phi}$.

    Consider the pairs $i, j$ such that $\phi_{i} \in \Phi', \phi_{j} \in \Phi'$.
    Since $\phi_{i}$ is a finite map, in particular it is surjective.
    Since $\phi_{i}$ is a projection of $\phi_{ij}$, we can deduce that $\overline{\phi_{ij}(V)}$ has dimension $1$.
    Further, $\overline{\phi_{ij}(V)}$ is irreducible since $V$ is irreducible.
    For all such $\phi_{ij}$, define $g_{ij}$ to be the generator of $I(\overline{\phi_{ij}(V)})$.
    Each $g_{ij}$ is a form of degree at most $D$, and is absolutely irreducible.
    The ideal $I$ is prime, and satisfies $\dim I = 1$.
    For every $i, j$ we can pick an invertible linear map $\bA^{n} \to \bA^{n}$ such that $\phi_{ij}$ is projection to the first two coordinates after applying the linear map.
    Therefore, we can apply \cref{lem:eliminationextendedeffective} to deduce that $g_{ij} \in A$, and we obtain $h_{ijk} \in A$, and $t_{ij} \in \bN$, and $a_{ij} \in \bZ\bs{\alpha}$ such that $a_{ij} g^{t_{ij}}_{ij} = \sum_{k} f_{k} h_{ijk}$.
    Further, $\deg{h_{i}}, t \leq 8d^{n^{2}}$, and $\lheight{a}, \lheight{h_{i}}, \lheight{g_{ij}} \leq h \cdot \log{m} \cdot e^{c_{1}} \cdot d^{c_{1}n^{c_{1}}}$.
    By \cref{thm:kaltofen mod p irreducible}, for each such $ij$ there exists an integer $\Delta_{ij}$ with $\lheight{\Delta_{ij}} \leq h \cdot e^{c_{2}} \cdot d^{cn^{c_{2}}}$ such that for any $p \not| \Delta_{ij}$ and $p \not | \disc{q}{z}$, and for any $\kp$ above $p$, the image of polynomial $g_{ij}$ under the map $\bZ\bs{\alpha} \to \bF_{p}\bs{z} / q_{1}$ remains absolutely irreducible.

    Define $\Delta' :=  \prod_{i,j} \resultant{a_{ij}}{q}{z} \cdot \Delta_{ij}$, where the product is over $i, j$ such that $\phi_{i}, \phi_{j} \in \Psi'$.
    Finally, define $\Delta := \Delta' \cdot \Delta'' \cdot (B+1)!$, where $\Delta''$ is the integer obtained by applying \cref{lem:dimpreserve} to $V$.
    We have $\lheight{\Delta'} \leq h \cdot e^{c_{2}} \cdot d^{cn^{c_{2}}} \cdot D^{(2n+2)} \leq h \cdot e^{c_{3}} \cdot 2^{\br{n \log \sigma}^{c_{3}}}$.
    Further $\lheight{\Delta''} \leq h \cdot e^{c_{4}} \cdot 2^{\br{n \log \sigma}^{c_{4}}}$, and $\lheight{(B+1)!} \leq d^{n^{3}}$.
    Therefore, $\lheight{\Delta} \leq h \cdot e^{c} \cdot 2^{\br{n \log \sigma}^{c}}$.

    Pick a prime $p \not| \Delta$ and $\kp$ above $p$, we will show that $V_{\kp}$ is irreducible, which will complete the proof.
    Since $p \not| \Delta''$, $V_{\kp}$ has dimension $1$.
    Further, since $p \not| (S+1)!$, we have $p > S+1$.
    This ensures that the set of maps in $\Phi$ are distinct, even when considered as maps from $\overline{\bF}_{p}^{n} \to \overline{\bF}_{p}^{2}$.
    
    Fix any $i, j$ such that $\phi_{i} \in \Phi', \phi_{j} \in \Phi'$.
    The equation $a_{ij} g^{t_{ij}}_{ij} = \sum_{k} f_{k} h_{ijk}$ holds in $\bF_{p}\bs{z}/q_{1}$, and $a_{ij} \neq 0$, therefore $g_{ij}^{t_{ij}} \in I_{\kp}$ for every such $\kp$.
    Further, since $p \not| \Delta_{ij}$, the polynomial $g_{ij}$ remains absolutely irreducible.
    These facts together imply that $\overline{\phi_{ij}(V_{\kp})}$ is either empty or irreducible, here we treat $\phi_{ij}$ as a map from $\overline{\bF}_{p}^{n} \to \overline{\bF}_{p}^{2}$.
    The total number of such maps is $(1 - 1/D)^{2} \abs{\Phi}^{2}$.

    Assume towards a contradiction that $V_{\kp}$ is reducible.
    Let $C_{1}, C_{2}$ be two components of $V_{\kp}$, we have $\deg(C_{1}), \deg(C_{2}) \leq D$.
    Define $\Phi_{\kp} \subset \Phi$ such that $\phi_{i} \in \Phi_{\kp}$ if $\phi_{i} : C_{1} \to \overline{\bF}_{p}^{1}$ and $\phi_{i} : C_{2} \to \overline{\bF}_{p}^{1}$ are both finite maps.
    Similar to the argument for $\Phi'$, we can deduce that $\abs{\Phi_{\kp}} \geq (1 - 2/D) \abs{\Phi}$ and we can deduce that for any $i, j$ with $\phi_{i} \in \Phi_{\kp}$ we have $\dim \overline{\phi_{ij}(C_{1})} = 1$ and $\dim \overline{\phi_{ij}(C_{2})} = 1$.
    We bound the number of maps $\phi_{ij}$ such that $\phi_{ij}(C_{1}) = \phi_{ij}(C_{2})$.
    Fix $i$, and pick a point $\alpha \in C_{1} \setminus C_{2}$.
    Since $\phi_{i}: C_{2} \to \overline{\bF}_{p}^{1}$ is finite, in particular it is surjective and therefore the curve $C_{2}$ is not contained in the hyperplane defined by $\phi_{i} = \phi_{i}(\alpha)$.
    Therefore, $C_{2} \cap \bc{\phi_{i} = \phi_{i}(\alpha)}$ is a finite set consisting of at most $D$ points, call these points $\beta_{1}, \dots, \beta_{D}$.
    The set of points $\alpha - \beta_{1}, \dots, \alpha - \beta_{D}$ is a finite set of nonzero points, therefore at most $(D / S) \abs{\Phi}$ many $\phi_{j}$ are such that $\phi_{j}$ is zero at one of these points.
    For every remaining $\phi_{j}$, the intersection of $C_{2}$ with the hyperplanes $\phi_{i} = \phi_{i}(\alpha)$ and $\phi_{j} = \phi_{j}(\alpha)$ is empty, whence $\phi_{ij}(C_{1}) \neq \phi_{ij}(C_{2})$.
    Since the maps $\phi_{ij}$ are finite on $C_{1}, C_{2}$, this implies $\overline{\phi_{ij}(C_{1})} \neq \overline{\phi_{ij}(C_{2})}$.
    Therefore, the total number of pairs $i, j$ such that $\phi_{i}, \phi_{j} \in \Phi_{\kp}$ and $\overline{\phi_{ij}(C_{1})} \neq \overline{\phi_{ij}(C_{2})}$ is at least $(1 - 2/D)(1 - 3/D) \abs{\Phi}^{2}$.
    Recall however that for at least $(1 - 1/D)^{2} \abs{\Phi}^{2}$ maps, $\overline{\phi_{ij}(V_{\kp})}$ was either irreducible or empty, in particular $\overline{\phi_{ij}(C_{1})} = \overline{\phi_{ij}(C_{2})}$.
    This is a contradiction to our assumption that $V_{\kp}$ is reducible, and completes the proof.
\end{proof}

The following result again assumes $f_{1}, \dots, f_{m} \in R$.
The proof will invoke \cref{lem:dimpreserve} and \cref{thm:irredpreserve} in the general setting on a new system of polynomials that extends $f_{1}, \dots, f_{m}$.
We note that this is a limitation of our proof, and we expect the result to also hold if $f_{1}, \dots, f_{m} \in A$.
\begin{theorem}
    \label{thm:redpreserve}
    Suppose $f_{1}, \dots, f_{m} \in R$.
    Suppose $\dim I = 1$, that $V$ is equidimensional, and that $V$ has at least two irreducible components.
    There exists an integer $\Delta$ with $\lheight{\Delta} \leq h \cdot 2^{\br{n \log \sigma}^{c}}$ and a polynomial $G \in \bZ\bs{z}$ with $\deg{G} \leq D^{4}$ and $\lheight{G} \leq h \cdot 2^{\br{n \log \sigma}^{c}}$ such that for any prime $p \not | \Delta$ and such that $G \pmod p$ has a root in $\bF_{p}$, the zeroset $V_{p}$ has at least two irreducible components that are $\bF_{p}$-definable.
\end{theorem}

\begin{proof}
    Let $k$ be the number of irreducible components of $V$, and suppose $V = V_{1} \cup \cdots V_{k}$ is the irreducible decomposition of $B$.
    By \bezout's theorem (\cref{theorem: bezout inequality}) we have $k \leq D := d^{n}$.
    As a first step, we will find defining equations for two components of $V$.
    \footnote{A reviewer pointed out to us that the existence of such equations also follows from \cite[Theorem~1.1]{scheiblechner2010generalization}.}
    The components of $V$ might not be $\bQ$-definable, therefore the defining equations we find will potentially lie in some algebraic extension of $\bQ$.

    Pick a linear map $\phi : \bA^{n} \to \bA^{2}$, where the coefficients of each coordinate map are picked from the set $\bs{D^{6}}$, such that $\pi|_{V} : V \to \phi(V)$ is finite, and such that each of the $k$ components of $V$ have distinct images under $\phi$.
    A random linear map has these properties, and therefore such a map exists.
    Since $\dim I = 1$ we have $\dim \phi(V) = 1$.
    Further, since $I$ is equidimensional, $\overline{\pi(V)}$ has exactly $k$ components that are each of codimension $1$, therefore $\overline{\phi(V)} = Z(g_{1})$ for some $g_{1}$.
    We can apply \cref{lem:eliminationextendedeffective} to deduce that $g_{1} \in R$, and the existence of polynomials $h_{1k} \in R$, and $a_{1}, t_{1} \in \bZ$ such that $a_{1} g_{1}^{t_{1}} = \sum_{k} f_{k} h_{1k}$.
    Further $\lheight{g_{1}}, \lheight{h_{1k}} \leq h \cdot d^{c_{1}n^{c_{1}}}$.

    Now $g_{1}$ has exactly $k$ absolutely irreducible factors, $g_{1} = g_{11} \cdot g_{12} \cdot \cdots \cdot g_{1k}$.
    Each of these factors are coprime, and occur with multiplicity $1$.
    Further, each factor $g_{1i}$ corresponds to the irreducible component $V_{i}$.
    Consider $I_{1} := I + g_{11}$ and $I_{2} := I + g_{12}$.
    We have $V_{1} \subset Z(I_{1})$, and $V_{2} \subset Z(I_{2})$.
    None of the components $V_{2}, \dots, V_{k}$ vanish on $g_{11}$.
    Therefore $\br{V_{2} \cup \cdots \cup V_{k}} \cap Z(g_{11})$ has dimension $0$.
    By Bézout's inequality, and the fact that $\deg g_{11} \leq \deg g_{1} \leq d^{n}$, we deduce that $\br{V_{2} \cup \cdots \cup V_{k}} \cap Z(g_{11})$ consists of at most $d^{2n}$ points.
    Similarly, $\br{V_{1} \cup V_{3} \cup \cdots \cup V_{k}} \cap Z(g_{12})$ consists of at most $d^{2n}$ points.

    In order to define $V_{1}$ (respectively $V_{2}$), we need to add equations to $I_{1}$ (respectively $I_{2}$) that vanish on $V_{1}$ (respectively $V_{2}$) but not on the above points.
    Towards this, pick a second linear map $\psi : \bA^{n} \to \bA^{2}$ where the coefficients of each coordinate map are picked from the set $\bs{D^{6}}$, such that $\psi|_{V} : V \to \psi(V)$ is finite, and such that each of the $k$ components of $V$ have distinct images under $\psi$.
    We further require that none of the $d^{2n}$ points of $\br{V_{2} \cup \cdots \cup V_{k}} \cap Z(g_{11})$ are mapped to $\psi(V_{1})$, and that none of the $d^{2n}$ points of $\br{V_{1} \cup V_{3} \cup \cdots \cup V_{k}} \cap Z(g_{12})$ are mapped to $\psi(V_{2})$.
    Such a map exists, since a random map satisfies these properties.
    As before, we have $g_{2}$ such that $\psi(V) = Z(g_{2})$, and by \cref{lem:eliminationextendedeffective} we have $g_{2} \in R$, and the existence of polynomials $h_{2k} \in R$, and $a_{1}, t_{1} \in \bZ$ such that $a_{2} g_{2}^{t_{2}} = \sum_{k} f_{k} h_{2k}$.
    Further $\lheight{g_{2}}, \lheight{h_{2k}} \leq h \cdot d^{c_{1}n^{c_{1}}}$.
    We can also factor $g_{2}$ as $g_{2} = g_{21} \cdot g_{22} \cdot \cdots \cdot g_{2k}$ with $g_{2i}$ corresponding to $V_{i}$.

    We now set $J_{1} := I_{1} + g_{21} = I + g_{11} + g_{21}$ and $J_{2} := I_{2} + g_{22} = I + g_{12} + g_{22}$.
    By construction, $Z(J_{1}) = V_{1}$ and $Z(J_{2}) = V_{2}$.

    The next step is to control the heights of $g_{11}, g_{12}, g_{21}, g_{22}$, in a number field that contains all these elements.
    To this end, we treat $g_{1}, g_{2}$ as polynomials in two variables each, say $y_{1}, y_{2}$.
    We can assume that the leading coefficients of $g_{1}, g_{2}$ in $y_{1}$ belong to $\bZ$: if not then this can be achieved by a random linear transformation, and we could have applied this transformation to the linear maps defining $g_{1}, g_{2}$ themselves.
    If $a$ is the leading coefficient of $g_{1}$, then we write $\widetilde{g}_{1} := a^{\deg{g_{1}} - 1} g(y_{1} / a, y_{2} + b)$, for a random $b$ with $\lheight{b} \leq \cdot 2^{\br{n \log \sigma}^{c_{2}}}$.
    After this transformation, the polynomial $\widetilde{g}_{1}$ is monic in $y_{1}$, and $\widetilde{g}_{1}(y_{1}, 0)$ is squarefree.
    Further, factors of $g_{1}$, in particular $g_{11}, g_{12}$ are in bijection with factors of $\widetilde{g}_{1}$.

    We now focus on $\widetilde{g}_{11}$, the factor corresponding to $g_{11}$ under this bijection.
    By \cref{lem:factorheight}, there is an irreducible factor $q_{11}(z)$ of $\widetilde{g}_{1}(z, 0)$, integers $\Delta_{11}, \Delta'_{11}$, and polynomials $u_{11}, v_{11}, w_{11} \in \bZ\bs{y_{1}, y_{2}, z}$ such that
    $$\Delta_{11} \Delta'_{11} \widetilde{g}_{1}(y_{1}, y_{2}) = u_{11}(y_{1}, y_{2}, z) v_{11}(y_{1}, y_{2}, z) + w_{11}(y_{1}, y_{2}, z) q_{11}(z).$$
    Further, $\deg_{z}(u_{11}), \deg_{z}(v_{11}) < e$.
    Finally, there is a root $\alpha_{11}$ of $q_{11}$ such that $g_{11} = u_{11}(x, y, \alpha_{11})$, potentially after scaling.
    By Gelfond's inequality \cite[B.7.3]{hindry2013diophantine} we have $\lheight{q_{11}} \leq h \cdot 2^{\br{n \log \sigma}^{c_{3}}}$, and we also have $\deg q_{11} \leq \deg g_{1} = D$.
    We also have $\lheight{\Delta_{11}}, \lheight{\Delta'_{11}}, \lheight{u_{11}} \leq h \cdot 2^{\br{n \log \sigma}^{c_{3}}}$ by \cref{lem:factorheight}.

    We similarly construct $q_{ij}, \alpha_{ij}, u_{ij}, v_{ij}, w_{ij}, \Delta_{ij}, \Delta'_{ij}$.
    We now construct an extension of $\bQ$ that contains every $\alpha_{ij}$.
    To this end, we construct a primitive element of $\alpha_{11}, \alpha_{12}, \alpha_{21}, \alpha_{22}$.
    Using \cref{theorem: koiran primitive element} we can deduce the existence of a polynomial $G$ with $\deg{G} \leq D^{4}$ and $\lheight{G} \leq h \cdot 2^{\br{n \log \sigma}^{c_{4}}}$, and a root $\gamma$ of $q$, along with $Q_{ij}, r$ such that $\alpha_{ij} = Q_{ij}(\gamma) / r$.
    We also have $\lheight{Q_{ij}, r} \leq h \cdot 2^{\br{n \log \sigma}^{c_{4}}}$.
    In equation 
    $$\Delta_{11} \Delta'_{11} \widetilde{g}_{1}(y_{1}, y_{2}) = u_{11}(y_{1}, y_{2}, z) v_{11}(y_{1}, y_{2}, z) + w_{11}(y_{1}, y_{2}, z) q_{11}(z)$$
    we substitute $z = Q_{11} / r$, multiply throughout by $r^{D^{2}}$, and reduce the coefficients mod $G$ to obtain
    $$\Delta_{11} \Delta'_{11} r^{D^{2}} \widetilde{g}_{1}(y_{1}, y_{2}) = \widetilde{u}_{11}(y_{1}, y_{2}, z) \widetilde{v}_{11}(y_{1}, y_{2}, z) + \widetilde{w}_{11}(y_{1}, y_{2}, z) G(z).$$
    We now have $\widetilde{u}(y_{1}, y_{2}, \gamma) = \widetilde{g}_{11}$, potentially after scaling.
    Up to scaling by $a$, the polynomial $g_{11}$ is exactly $\widetilde{u}(ay_{1}, y_{2} + b, \gamma)$.
    This shows that $\lheight{g_{11}} \leq h \cdot 2^{\br{n \log \sigma}^{c_{5}}}$.
    Therefore, we have obtained the height bounds on $g_{11}$ that we were looking for, and the same bound holds for $g_{ij}$.

    Recall that $J_{1}, J_{2}$ correspond to different irreducible components of $I$, and that $J_{1} + J_{2}$ has dimension $0$.
    Pick any prime $p \not | \disc{z}{G}$ such that $G \pmod p$ has a linear factor.
    Let $\kp$ be a prime of the ring of integers of $\bQ\bs{z} / G(z)$ that corresponds to the linear factor.
    The quotient map $\bZ\bs{\gamma} \to \bZ\bs{\gamma} / \kp$ has image in $\bF_{p}$.
    We now apply \cref{lem:dimpreserve} and \cref{thm:irredpreserve} to the ideals $J_{1}, J_{2}$, and $J_{1} + J_{2}$.
    We deduce that there is an integer $\Delta_{0}$ such that for any prime $p \not | \Delta_{0}$, the images of $J_{1}, J_{2}$ remain irreducible, and the dimensions of $I, J_{1}, J_{2}, J_{1} + J_{2}$ are one, one, one, zero respectively.
    We also have $Z((J_{1})_{\kp}) \subset Z(I_{\kp})$.
    Therefore, $Z((J_{1})_{\kp})$ is an $\bF_{p}$-definable irreducible component of $Z(I_{\kp})$.
    The same holds for $Z((J_{2})_{\kp})$.
    Since $\dim (J_{1} + J_{2})_{\kp} = 0$, these two components are different from each other.
\end{proof}

%% file: interactive-proofs.tex
\RO{clean this up a bit later}

In this section, we prove our main theorem: interactive protocols for testing non-primality of natural classes of ideals.
As we have discussed in \cref{section: introduction}, we first give some protocols for certain special cases, as the protocol in the main theorem will invoke these special cases as subroutines.

The setting of this section is that of our main theorem: we assume that we have an algebraic circuit $C'$ of size $5sm$ with integer coefficients that computes polynomials $f_{1}, \dots, f_{m} \in \bZ\bs{x_{1}, \dots, x_{m}}$, and the partial derivatives $\partial_{i} f_{j}$.
The existence of $C'$ is guaranteed by \cref{lemma: derivatives of a circuit} applied to the input circuit $C$.
By \cref{lemma: circuit height bound} also applied to the input circuit $C$, we have $\lheight{f_{i}} \leq 2^{2s}$ and $\deg{f_{i}} \leq 2^{s}$.
Define $d := \max \deg{f_{i}}$.
Define $\sigma := dm + 2$.

Let $I := \ideal{f_{1}, \dots, f_{m}}$, and let $V := Z(I)$ denote the algebraic set corresponding to $I$.
For any prime $p$, define $I_{p}$ to be the ideal in $\bF_{p}\bs{x_{1}, \dots, x_{n}}$ generated by the images of $f_{1}, \dots, f_{m}$ under the quotient map $R \to \bF_{p}\bs{x_{1}, \dots, x_{n}}$, and let $V_{p}$ denote the algebraic set of $I_{p}$ in $\overline{\bF}_{p}\bs{x_{1}, \dots, x_{n}}$.

We begin with a protocol to decide if a zero dimensional ideal has at least two points in its zeroset.

\begin{lemma}
    \label{lem:zero dimensional reducible protocol}
    Assume GRH.
    If $I$ is promised to be zero dimensional, then there is an AM protocol that verifies $\abs{V} \geq 2$.
\end{lemma}
\begin{proof}
    The proof uses \cref{cor:onerootheightbound}, \cref{lem:tworootsbounds} and the Goldwasser-Sipser protocol (\cref{lem:goldwassersipser}).
    Let $h := \max_{i} \lheight{f_{i}}$, by \cref{lemma: circuit height bound} we have $h \leq 2^{2s}$.

    Define $\Pi_{red}(x)$ to be the set of primes $p \leq x$ such that $V_{p}$ has at least two points in $\bF_{p}$, define $\pi_{red}(x) := \abs{\Pi_{red}(x)}$.
    By \cref{cor:onerootheightbound}, if $\abs{V} = 1$ then $\pi_{red}(x) \leq h \cdot 2^{\br{n \log \sigma}^{c_{1}}}$ for all $x$.
    Suppose $\abs{V} \geq 2$.
    Let $G$ be the polynomial guaranteed by \cref{lem:tworootsbounds}.
    By the effective Chebotarev density theorem (\cref{lem:effective chebotarev}) we have 
    $$\pi_{G}(x) \geq \frac{1}{\deg G} \br{\pi(x) - \log \Delta_{G} - c_{2} x^{1/2} \log\br{\Delta_{G} x^{\deg(G)}}},$$
    where $\Delta_{G}$ is the discriminant of $G$, and $\pi_{G}(x)$ is the number of primes $p \leq x$ such that $G$ has a root in $\bF_{p}$.
    By \cref{lem:tworootsbounds} $\deg G \leq 2^{\br{n \log \sigma}^{c_{3}}}$, and $\lheight{G} \leq h \cdot 2^{\br{n \log \sigma}^{c_{4}}}$.
    Bounding the discriminant using \cref{cor:resdischeightbounds}, we have
    $$\pi_{G}(x) \geq \frac{\pi(x)}{2^{\br{n \log \sigma}^{c_{3}}}} - c_{5} h x^{1/2} 2^{\br{n \log \sigma}^{c_{6}}} - c_{7} x^{1/2} \log x.$$
    By \cref{lem:tworootsbounds} we have $\pi_{red}(x) \geq \pi_{G}(x) - h \cdot 2^{\br{n \log \sigma}^{c_{8}}}$.
    Finally, by \cref{thm:unconditional prime counting} we have $\pi(x) > \frac{x}{\log_{e}{x}}$ for all $x$ larger than $17$.
    For $x_{0} = h^{c_{9}} \cdot 2^{\br{n \log \sigma}^{c_{10}}}$ for a large enough constants $c_{9}, c_{10}$, we see that $\pi_{red}(x_{0}) \geq 2 h \cdot 2^{\br{n \log \sigma}^{c_{1}}}$.
    Note that $\lheight{x_{0}}$ is polynomial in the input.

    We can now invoke the Goldwasser-Sipser protocol (\cref{lem:goldwassersipser}) on the set $\Pi_{red}(x_{0})$, with $K = h \cdot 2^{\br{n \log \sigma}^{c_{1}}}$.
    Membership in this set is in NP: the certificate for membership are two distinct roots, which has polynomial bit complexity, and we can evaluate the input polynomials on these roots to verify them.
    As shown above, $\abs{V} \geq 2$ if and only if $\pi_{red}(x_{0}) \geq 2K$ and $\abs{V} = 1$ if and only if $\pi_{red}(x_{0}) \leq K$.
    Therefore, the protocol correctly verifies $\abs{V} \geq 2$.
\end{proof}

The second special case is a protocol that decides if a zeroset has at least two components of highest dimension.

\begin{lemma}
    \label{lem:twotopcomponents}
    Assume GRH.
    If $I$ is promised to be equidimensional and have dimension $r$, then there is an AM protocol that verifies that $V$ has at least two irreducible components of dimension $r$.
\end{lemma}
\begin{proof}
    Let $D := d^{n}$.
    Apply the algorithm of \cref{cor:effectivebertini} in order to obtain a linear subspace $H$.
    The space $H$ is defined by linear equations $h_{1}, \dots, h_{r-1}$, with $\lheight{h_{i}} \leq \br{nD}^{c_{1}}$.
    Let $J := I + \ideal{h_{1}, \dots, h_{r-1}}$, and let $W := Z(J)$.
    By \cref{cor:effectivebertini}, with high probability we have $\dim(W) = 1$, and the number of irreducible components of dimension $1$ of $W$ is the same as the number of irreducible components of dimension $r$ in $V$.
    For the rest of the proof, we assume that this event occurs, this will only change the soundness and completeness of our protocol by a small constant, which does not change the complexity class AM.
    Let $J_{p}, W_{p}$ be defined the same way as $I_{p}, V_{p}$.

    Let $h := \max\br{\max_{j} \lheight{h_{j}},\max_{i} \lheight{f_{i}}}$, we have $h \leq \br{nD}^{c_{1}} + 2^{2s}$.
    Let $\Pi_{red}(x)$ be the set of primes $(150)^{2} D^{8} < p < x$ such that $\abs{W_{p} \cap \bF_{p}^{n}} \geq 2p - 10 D^{4} p^{1/2}$, and let $\pi_{red}(x) := \abs{\Pi_{red}(x)}$.
    By an effective Lang-Weil bound (\cref{thm:effective lang weil}), $p \in \Pi_{red}(x)$ if and only if $(150)^{2} D^{8} < p < x$ and $W_{p}$ has at least two irreducible components of dimension $1$ defined over $\bF_{p}$.

    Suppose $W$ has only one irreducible component of dimension $1$ (this component is then automatically defined over $\bF_{p}$).
    By \cref{thm:irredpreserve}, for all $x$ we have $\pi_{red}(x) \leq h \cdot 2^{\br{n \log \sigma}^{c}}$, since for all but $h \cdot 2^{\br{n \log \sigma}^{c}}$ primes $W_{p}$ has at most one irreducible component defined over $\bF_{p}$.
    On the other hand, suppose $W$ has at least two irreducible components of dimension $1$.
    Let $G$ be the polynomial guaranteed by \cref{thm:redpreserve}.
    By the effective Chebotarv density theorem \cref{lem:effective chebotarev} we have
    $$\pi_{G}(x) \geq \frac{\pi(x)}{2^{\br{n \log \sigma}^{c_{2}}}} - c_{3} h x^{1/2} 2^{\br{n \log \sigma}^{c_{4}}} - c_{5} x^{1/2} \log x.$$
    By \cref{thm:redpreserve}, for all but $h \cdot 2^{\br{n \log \sigma}^{c_{6}}}$ primes in $\pi_{G}(x)$, the zeroset $W_{p}$ has at least two components, therefore every such prime larger than $(150)^{2} D^{8}$ is in $\Pi_{red}(x)$.
    For $x_{0} = h^{c_{7}} \cdot 2^{\br{n \log \sigma}^{c_{8}}}$ for a large enough constants $c_{7}, c_{8}$, we see that $\pi_{red}(x_{0}) \geq 2 h \cdot 2^{\br{n \log \sigma}^{c}}$.

    We can now invoke the extension of the Goldwasser-Sipser protocol (\cref{cor:goldwasser sipser extended}) on the set $\Pi_{red}(x_{0})$, with $K = h \cdot 2^{\br{n \log \sigma}^{c}}$.
    Membership in this set is itself in AM: note that $\abs{W_{p} \cap \bF_{p}^{n}}$ is either at most $p + 5 D^{4} p^{1/2} + D$ or at least $2p - 10 D^{4} p^{1/2} - D$.
    Since $p > (150)^{2} D^{8}$ the ratio of these sizes is at least $1.9$.
    If the protocol accepts, then with high probability $W$, and therefore $V$ has at least two irreducible components of top dimension.
    If not, then with high probability $V$ has only one component of top dimension.
\end{proof}

\begin{remark}
\label{remark: reducible base change tight}
    We give an example that shows that \cref{thm:redpreserve} is tight in some sense.
    The theorem, combined with the effective Chebotarev density theorem \cref{lem:effective chebotarev} shows that the density of primes for which $\bF_{p}$ has two $\bF_{p}$ definable components is inverse exponential.
    We give an example to show that this inverse exponential behaviour is unavoidable.
    
    Consider $I = \ideal{x_{1}^{2} - 2, x_{2}^{2} - 3, \dots, x_{n-1}^{2} - p_{n-1}}$, where $p_{j}$ is the $j^{th}$ prime.
    The zeroset $V$ consists of $2^{n-1}$ components, each of dimension $1$.
    Similarly, for any prime $p > p_{n-1}$, the zeroset $V_{p}$ also consists of $2^{n-1}$ components, each of dimension $1$.
    However, these components might not be defined over $\bF_{p}$ itself.
    For any prime $p$, the components will be defined over the smallest extension of $\bF_{p}$ that contains square roots of $2, 3, \cdots, p_{n-1}$.
    The only primes $p$ for which there exist $\bF_{p}$ definable components of $V_{p}$ are those where $2, 3, \dots, p_{n-1}$ are all quadratic residues.
    By quadratic reciprocity and Dirichlet's theorem, the density of such primes within all primes is $2^{-n+1}$.
\end{remark}

\subsection{Interactive proof for radical ideals}\label{subsection: interactive  radical}

We now give an AM protocol for deciding non primality of radical ideals, using the above two protocols as subroutines.

\begin{algorithm}
  \caption{AM protocol for non-primality of radical ideals
    \label{alg:radical non prime}}
    \SetKwInOut{Input}{Input}
    \SetKwInOut{Output}{Output}
    \SetKwInOut{Merlin}{Merlin}
    \SetKwInOut{Arthur}{Arthur}

    \Input{A circuit $C'$ that computes polynomials $f_{1}, \dots, f_{m} \in R$ along with the derivatives $\partial_{i} f_{j}$, with $\deg{f_{i}} \leq d$ such that $I := \ideal{f_{1}, \dots, f_{m}}$ is a radical ideal, and the dimension $r$ of $I$.
    }

    \Arthur{If $r = 0$, then perform the protocol in \cref{lem:zero dimensional reducible protocol}, and accept if and only if the protocol accepts.
    If $r > 0$, then send the empty string to Merlin.
    }

    \Merlin{Compute the Jacobian matrix $\cJ_{ij} = \partial f_{i} / \partial x_{j}$.
    Check if there is a minor $M$ of $\cJ$ of size at least $n-r+1$ such that $M \neq 0, f_{1} = 0, \dots, f_{m} = 0$ is a satisfiable system.
    If such a minor exists, send the rows and columns that correspond to $M$ to Arthur.
    If not, then send the empty string to Arthur.
    }

    \Arthur{
    If Merlin sends the empty string then perform the protocol in \cref{lem:twotopcomponents}, and accept if and only if the protocol accepts.
    If not, then construct an algebraic circuit for $1 - yM$, where $M$ is the minor of $\cJ$ corresponding to the rows and columns received from Merlin.
    Perform the protocol in \cref{thm: koiran nullstellensatz in am} on the inputs $f_{1}, \dots, f_{m}, 1 - yM$.
    Accept if and only if this protocol accepts.
    }
\end{algorithm}

\begin{theorem}
    \label{theorem: radical ideal protocol}
    Assume GRH.
    If $I$ is a radical ideal of dimension $r$ then \cref{alg:radical non prime} is a valid AM protocol for deciding $I$ is not prime.
\end{theorem}

\begin{proof}
    Suppose $I$ is not prime.
    Since $I$ is radical, it has at least two irreducible components.
    If $\dim I = 0$, then $\abs{V} \geq 2$, and the protocol in \cref{lem:zero dimensional reducible protocol} correctly accepts with high probability, and therefore \cref{alg:radical non prime} correctly accepts with high probability in the first step.
    If $r > 0$, then either $V$  has at least one component of dimension at most $r-1$, or $V$ is equidimensional and has at least two components of dimension $r$.
    If $V$ has a component of dimension less than $r$, then by \cref{lem:tangent dimension}, there is a point $x_{0} \in V$ such that $\Rank{\cJ(x_{0})} > n-r$.
    Therefore Merlin can always find a minor of $\cJ$ of size greater than $n-r$ that makes the polynomial system created  by Arthur in round 3 satisfiable. 
    Observe that the system has degree at most $nd+1$, and logarithmic height at most $h \cdot (dn)^{c}$, where $h := \max_{i} \lheight{f_{i}}$.
    Further, given the circuit $C'$ that computes $f_{1}, \dots, f_{m}$ and $\partial_{j} f_{i}$, Arthur can easily produce a circuit $C''$ that also computes $1 - yM$.
    In this case, the protocol in \cref{thm: koiran nullstellensatz in am} correctly accepts with high probability.
    If $V$ has two components of dimension $r$, then the protocol in \cref{lem:twotopcomponents} correctly accepts with high probability.

    Suppose $I$ is a prime ideal.
    If $r = 0$, then $\abs{V} = 1$, and the protocol in \cref{lem:zero dimensional reducible protocol} correctly rejects with high probability, and therefore \cref{alg:radical non prime} correctly fails to accept with high probability in the first step.
    If $r \geq 0$, then $V$ has only one irreducible component.
    The Jacobian $\cJ$ has rank at most $n-r$ at every point in $V$, therefore no matter what choice of rows and columns Merlin picks in round $2$, the system created by Arthur in round $3$ is unsatisfiable, and therefore the protocol in \cref{thm: koiran nullstellensatz in am} fails to accept with high probability.
    If Merlin sends the empty string, then the protocol in \cref{lem:twotopcomponents} fails to accept with high probability.
    Therefore, when $I$ is a prime ideal, the protocol in \cref{alg:radical non prime} correctly fails to accept with high probability.

    The above protocol is in $\AM[4]$, the theorem now follows from the fact that $\AM[4] = \AM$.
\end{proof}

\subsection{Equidimensional Cohen-Macaulay ideals}\label{subsection: interactive CM}

We now give an AM protocol for deciding non primality of equidimensional Cohen-Macaulay ideals.

\begin{algorithm}[H]
  \caption{AM protocol for non-primality of equidimensional CM ideals
    \label{alg:cm non prime}}
    \SetKwInOut{Input}{Input}
    \SetKwInOut{Output}{Output}
    \SetKwInOut{Merlin}{Merlin}
    \SetKwInOut{Arthur}{Arthur}

    \Input{A circuit $C'$ that computes polynomials $f_{1}, \dots, f_{m} \in R$ along with the derivatives $\partial_{i} f_{j}$, with $\deg{f_{i}} \leq d$ such that $I := \ideal{f_{1}, \dots, f_{m}}$ is an equidimensional Cohen-Macaulay ideal, and the dimension $r$ of $I$.
    }

    \Arthur{If $r = 0$, then perform the protocol in \cref{lem:zero dimensional reducible protocol}, and accept if the protocol accepts. \\
    If $r > 0$, then perform the protocol in \cref{lem:twotopcomponents}, and accept if the protocol accepts. \\
    If the above protocols fail to accept, let $Y, Z$ be $(n-r) \times m$ and $n \times (n-r)$ symbolic matrices (in new variables), and perform the protocol in \cref{thm: koiran dimension in AM} on the system $f_{1}, \dots, f_{m}, \det\br{Y\cJ Z}$, with parameter $(n-r)(m+n) + r$.
    Accept if and only if this protocol accepts.
    }
\end{algorithm}

\begin{theorem}
    \label{theorem: CM ideal protocol}
    Assume GRH.
    If $I$ is a equidimensinal Cohen-Macaulay ideal of dimension $r$ then \cref{alg:cm non prime} is a valid AM protocol for deciding $I$ is not prime.
\end{theorem}

\begin{proof}
    Suppose $I$ is not prime.
    Then either $I$ has at least two components of dimension $r$, or $I$ has exactly one irreducible component but is not radical.
    If $I$ has at least two irreducible components of dimension $r$, then depending on $r$ either the protocol in \cref{lem:zero dimensional reducible protocol} or the protocol in \cref{lem:twotopcomponents} correctly accepts with high probability.
    Suppose now that $I$ has a unique minimal prime, call this prime $\kp$.
    Since $S$ is an affine domain of dimension $n$, by \cite[Cor~13.4]{eisenbud2013commutative} we have $\codim I = \codim \kp = n - r$.
    Let $J$ be the ideal generated by the $(n-r) \times (n-r)$ minors of $\cJ$.
    Since $S / I$ is Cohen-Macaulay and $I$ is not radical, by \cref{theorem: serre jacobian} we deduce that $J$ has codimension $0$ in $S / I$, equivalently every generator of $J$ lies in $\kp$.
    The polynomial $\det\br{Y\cJ Z}$ lies in the ideal $J \tensor \overline{\bQ}\bs{Y, Z}$, therefore $\det\br{Y\cJ Z} \in \kp \tensor \overline{\bQ}\bs{Y, Z}$.
    The ideal $\kp \tensor \overline{\bQ}\bs{Y, Z}$ has dimension $r + (n-r)(m+n)$, and is the unique minimal prime of $I \tensor \overline{\bQ}\bs{Y, Z}$, therefore $\dim \det\br{Y\cJ Z} + I \tensor \overline{\bQ}\bs{Y, Z} = r + (n-r)(m+n)$.
    Therefore the protocol in \cref{thm: koiran dimension in AM} correctly accepts with high probability.

    Now suppose $I$ is prime.
    Since $I$ has only one irreducible component, depending on $r$ either the protocol in \cref{lem:zero dimensional reducible protocol} or the protocol in \cref{lem:twotopcomponents} fails to accept with high probability.
    Since $S/I$ is reduced, \cref{theorem: serre jacobian} implies that $J$ has codimension at least $1$ in $S / I$.
    Equivalently, there is at least one generator of $J$ that does not lie in the ideal $I$.
    Therefore, $\det\br{Y\cJ Z} \not \in I \tensor \overline{\bQ}\bs{Y, Z}$, and $\dim \det\br{Y\cJ Z} + I \tensor \overline{\bQ}\bs{Y, Z} = r + (n-r)(m+n) - 1$.
    The protocol in \cref{thm: koiran dimension in AM} correctly fails to accept with high probability.
\end{proof}

\subsection{Proof of main theorems}

We now use the above protocols to prove our main theorems, which we restate here for convenience.

\interactive*

\begin{proof}
    Since we know the dimension of $I$, by \cref{theorem: radical ideal protocol} and \cref{theorem: CM ideal protocol} (depending on whether $I$ is radical or equidimensional Cohen-Macaulay), the complexity of deciding if $I$ is not prime lies in $\AM$, therefore the complexity of deciding if $I$ is prime lies in $\coAM$.
\end{proof}

\primesPH*

\begin{proof}
    We first compute the dimension of $I$.
    By \cref{thm: koiran dimension in AM}, there exists an $\AM$ protocol for checking if $\dim I \geq r$ for any $r$, and since $\AM \subset \Pi_{2}^{P}$, the complexity of computing the dimension of $I$ exactly lies in $P^{\Pi_{2}^{p}} \subset \Sigma_{3}^{p} \cap \Pi_{3}^{p}$.
    With the dimension at hand, we can apply \cref{theorem: interactive protocols for primality} and decide whether $I$ is prime in $P^{\Pi_2^p}$.
    Thus, the complexity of deciding if $I$ is prime is in $\Sigma_3^p \cap \Pi_3^p$.
\end{proof}

%% file: conclusion.tex
In this work, we proved that the ideal primality testing problem is in the third level of the polynomial hierarchy for the natural classes of ideals comprised of radical ideals and equidimensional Cohen-Macaulay ideals.
This significantly tightens the complexity-theoretic gap for the primality testing problem for these classes of ideals, given that the primality testing problem is already $\coNP$-hard for such classes.
Prior to our work, the best upper bounds for testing whether a radical ideal or a zero-dimensional ideal is prime was $\PSPACE$, whereas for complete intersections the best upper bound was $\EXP$.

\begin{itemize}
\item 
A common nice feature of the two classes of ideals that we studied is that any associated prime must be a minimal prime, and therefore we avoid the issue of having to detect {\em embedded} primes. 
The issue of embedded primes, which causes non-reduced behavior in "very small" parts of our algebraic set, is in general a very hard problem to identify (and therefore to handle). 
We leave it as an open question, to detect embedded primes either "via zeros" or algorithmically better than EXPSPACE.

\item 
We studied the problems over the field of $\bC$ (or $\bQ$). This allows us to go modulo several primes $p$, leading to the application of arithmetic-geometry results over finite fields (e.g.~Lang-Weil points count and Chebotarev primes count). The ideal primality testing problem over an input field $\overline{\bF}_q$, for prime $q$, rules out the idea of going modulo other primes $p\ne q$. This makes it an open question to extend Koiran's result \cite{K96} to input field $\overline{\bF}_q$. Similarly, our paper leaves it as an open question, to improve ideal primality testing beyond EXPSPACE, for natural classes of ideals over an algebraically closed field of positive characteristic. (E.g.~what about the dimension $\le1$ case here?)

\end{itemize}